\documentclass[10pt,reqno]{amsart}
\usepackage{ifxetex,ifluatex}
\newif\ifxetexorluatex
\ifxetex
  \xetexorluatextrue
\else
  \ifluatex
    \xetexorluatextrue
  \else
    \xetexorluatexfalse
  \fi
\fi
\ifxetexorluatex
  \usepackage{fontspec}
\else
  \usepackage[T1]{fontenc}
  \usepackage[utf8]{inputenc}
\fi
\usepackage{amsthm}
\usepackage{amsmath}
\usepackage{amssymb,booktabs} 
\usepackage{geometry}[a4paper,top=1.5cm,bottom=1.5cm,left=2cm,right=2cm]
\calclayout

\usepackage{
amsfonts, 
latexsym, 
enumitem,
}

\usepackage[english]{babel}
\usepackage{csquotes}
\numberwithin{equation}{section}
\usepackage{hyperref}
\usepackage{cleveref} 
\crefname{equation}{}{}
\usepackage{pgfplots,tikz}
\usetikzlibrary{external}
\pgfplotsset{compat=1.16}
\usepgfplotslibrary{fillbetween}
\definecolor{col1}{HTML}{d7191c}
\definecolor{col2}{HTML}{fdae61}
\definecolor{col3}{HTML}{abd9e9} 
\definecolor{col4}{HTML}{2c7bb6}
\ifdefined\C
\renewcommand{\C}{\mathbf{C}}
\else
\newcommand{\C}{\mathbf{C}}
\fi
\newcommand{\wh}{\widehat}

\DeclareMathOperator{\Tr}{Tr}

\DeclareMathOperator{\Cov}{\mathbf{Cov}}
\DeclareMathOperator{\Std}{\mathbf{Std}}
\DeclareMathOperator{\SFF}{\mathbf{SFF}}
\DeclareMathOperator{\Var}{\mathbf{Var}}

\newcommand{\ov}{\overline}
\makeindex
\newtheorem{theorem}{Theorem}[section]

\newtheorem{lemma}[theorem]{Lemma}
\newtheorem{proposition}[theorem]{Proposition}
\newtheorem{definition}[theorem]{Definition}

\newtheorem{remark}[theorem]{Remark}
\newtheorem{corollary}[theorem]{Corollary}

\newcommand{\ii}{\mathrm{i}}
\usepackage{bm}

\newcommand{\norm}[1]{\lVert #1\rVert}

\newcommand{\un}{\underline}

\newcommand{\wt}{\widetilde}
\newcommand{\cE}{\mathcal E}
\newcommand{\cO}{\mathcal O}

\newcommand{\E}{\mathbf E}
\newcommand{\Prob}{\mathbf P}
\newcommand{\R}{\mathbf R}

\newcommand{\N}{\mathbf N}

\newcommand{\dif}{\ensuremath{\operatorname{d}\!{}}}

\newcommand{\be}{\begin{equation}}
\newcommand{\ee}{\end{equation}}

\usepackage{mathtools}
\DeclarePairedDelimiter{\abs}{\lvert}{\rvert}%

\DeclarePairedDelimiter{\braket}{\langle}{\rangle}
\DeclarePairedDelimiterXPP{\landauO}[1]{\cO}(){}{#1}
\DeclarePairedDelimiterXPP{\landauo}[1]{\co}(){}{#1}
\DeclarePairedDelimiterX{\set}[1]\{\}{#1}
\crefrangeformat{equation}{(#3#1#4)-(#5#2#6)}
\crefrangeformat{theorem}{Theorems #3#1#4--#5#2#6}

\author{Giorgio Cipolloni}
\address{Princeton Center for Theoretical Science, Princeton University, Princeton, NJ 08544, USA}
\author{L\'aszl\'o Erd\H{o}s\(^{\dagger}\)}
\address{IST Austria, Am Campus 1, 3400 Klosterneuburg, Austria}
\author{Dominik Schr\"oder\(^{\ddagger}\)}
\address{Institute for Theoretical Studies, ETH Zurich, Clausiusstr.\ 47, 8092 Zurich, Switzerland}
\email{gc4233@princeton.edu} 
\email{lerdos@ist.ac.at}
\email{dschroeder@ethz.ch}
\thanks{$^\dagger$Partially supported by ERC Advanced Grant "RMTBeyond" No.~101020331}
\thanks{\(^{\ddagger}\)Supported by Dr.\ Max R\"ossler, the Walter Haefner Foundation and the ETH Z\"urich Foundation}
\subjclass[2010]{60B20} 
\keywords{Wigner-Dyson universality, GUE ensemble, Slope-dip-ramp-plateau}
\title{On the spectral form factor for random matrices}  
\date{\today}

\begin{document}

\begin{abstract}  
In the physics literature the spectral form factor (SFF), the squared Fourier transform of the empirical eigenvalue density, 
is the most common tool to test universality for disordered quantum systems, yet  previous
mathematical results have been restricted  only to two exactly
solvable models~\cite{MR4257834, MR4312363}. We rigorously prove the physics prediction on SFF
up to an intermediate time scale for 
a large class of random matrices using a robust method, the  multi-resolvent local laws. 
Beyond Wigner matrices we also consider the monoparametric ensemble and prove that
universality of SFF can already be triggered by a single random parameter,
supplementing the recently proven  Wigner-Dyson universality~\cite{2106.10200} to  larger spectral scales.
Remarkably, extensive numerics indicates that 
our formulas  correctly predict the SFF in the entire \emph{slope-dip-ramp} regime, as customarily 
called in physics.  

\end{abstract}

\maketitle

\section{Introduction}

Spectral statistics of disordered quantum systems tend to exhibit universal behavior  and
hence are widely used to  study quantum chaos and to 
identify universality classes. In the chaotic regime, the celebrated Wigner-Dyson-Mehta
eigenvalue gap statistics involving the well-known \emph{sine-kernel}~\cite{MR2129906} tests 
this universality  on the scale of 
individual eigenvalue spacing. On this small \emph{microscopic} scale the universality phenomenon is the most robust
and it depends only on the fundamental symmetry type of the model.
On larger scales more details of the model influence the spectral statistics, nevertheless
several qualitative and also quantitative  universal patterns still prevail.

\subsection{The spectral form factor and predictions from physics}

In the physics literature the standard tool to investigate eigenvalues $\lambda_1, \lambda_2, \ldots, \lambda_N$
of a Hermitian $N\times N$ matrix (Hamiltonian) $H$ on all scales at once is the \emph{spectral form factor}
(SFF) \cite{10032995}  defined as
\begin{equation}\label{def:sff1}
\SFF(t): = \frac{1}{N^2} \sum_{i,j=1}^N e^{\ii t(\lambda_i-\lambda_j)}=\abs{\braket{e^{\ii t H}}}^2
\end{equation}
with a real time parameter $t>0$, i.e.\ it is the square of the Fourier transform of the empirical spectral density.
Here we denoted the normalized trace of any $N\times N$ matrix $A$
by $\braket{A}= \frac{1}{N}\Tr A$.
In case of random $H$, the expectation of $\SFF(t)$ is denoted by
\begin{equation}\label{def:K}
    K(t) :=\E \big[ \SFF(t)\big].
\end{equation}
For typical disordered Hamiltonians
a key feature of $\SFF(t)$ is that for larger $t$ (more precisely, in the \emph{ramp} and \emph{plateau} regimes,
see later) it is strongly dependent on the sample, i.e.\ the standard deviation of $\SFF(t)$ is 
comparable with $K(t)$. In other words, $\SFF(t)$
 is  \emph{not} self-averaging~\cite{PhysRevLett.78.2280}
despite the large summation in~\eqref{def:sff1}.  

The spectral form factor and its expectation $K(t)$ have
 a very rich physics literature since they contain most physically relevant information about
spectral statistics. Quantizations of integrable systems typically result in $K(t)\sim 1/N$ for 
all $t$ where $N$ is the dimension of the Hilbert space.
Chaotic systems give rise to a linearly growing behavior of $K(t)$ for smaller  $t$  (so-called \emph{ramp}) until
it turns into  a flat regime, the \emph{plateau}. The turning point is around the Heisenberg time $T_H$, but the details
of the transition depend on the symmetry class of $H$ and on whether the eigenvalues are rescaled 
to take into account the non-constant density of states (in physics terminology: \emph{unfolding the spectrum}). 
For example,  in the time irreversible case (GUE symmetry class) the unfolded SFF
has a sharp kink, while in the GOE symmetry  class the kink is smoothened. The exact formulas can be 
computed  from the Fourier transform of the two point eigenvalue correlation function 
of the corresponding Gaussian random matrix ensemble, see~\cite[Eqs. (6.2.17), (7.2.46)]{MR2129906},  the result is
  \begin{equation}\label{KGUE}
    K_\mathrm{GUE}(\tau T_H) \approx \frac{1}{N}\times \begin{cases} \tau, &  0< \tau\le 1 \\ 1, &  \tau\ge 1 \end{cases},
    \qquad K_\mathrm{GOE}(\tau T_H) \approx \frac{1}{N}\times \begin{cases} 2\tau- \tau\log(1+2\tau), &  0< \tau\le 1 \\ 
    2-  \tau\log \frac{2\tau+1}{2\tau-1},&  \tau\ge 1 \end{cases},
\end{equation}
for any fixed $\tau>0$ in the large $N$ limit.
Here we expressed the physical time $t$ in units of the Heisenberg time, $\tau= t/T_H$, 
where $T_H$ is  given by $T_H = 2\pi\bar \rho$ with $\bar\rho$ being the
average density.  Choosing the standard normalisation  for the independent (up to symmetry) matrix elements,
\begin{equation}\label{Hnorm}
 \E h_{ij}=0, \qquad \E |h_{ij}|^2 =\frac{1}{N},
\end{equation}
 the limiting density of states is 
the semicircle law $\rho_{\mathrm{sc}}(E)=\frac{1}{2\pi}\sqrt{(4-E^2)_+}$, so  we have $N$ eigenvalues
in an interval of size 4, hence $\bar\rho = N/4$ and thus $T_H= \frac{\pi}{2}N$.  In particular, 
in the original $t$ variable 
\begin{equation}\label{KGUEt}
 K_\mathrm{GUE}(t) \approx \begin{cases} \frac{2t}{\pi N^2} , &  \delta N\le t\le  \frac{\pi}{2}N  \\ \frac{1}{N}, 
 &  t\ge  \frac{\pi}{2}N.  \end{cases}
\end{equation}
Note the lower bound on $t$: the formula holds in the large $N$ limit in the regime where 
$t\ge \delta N$  for some fixed $\delta>0$ that is independent of $N$. 
The corresponding formulas without unfolding the spectrum (i.e.\ for the quantity defined in~\eqref{def:sff1}) 
are somewhat different, see e.g. \cite[Eq. (4.8)]{cond-mat/9608116} for the GUE case;  they still 
have a ramp-plateau shape but the kink is smoothened.

The ramp-plateau picture and its sensitivity to the symmetry type
has  been established well beyond the standard mean field   random matrix models.
In fact, the  Bohigas-Giannoni-Schmit conjecture \cite{MR730191}   asserts that the formulas
\eqref{KGUE} are universal, i.e.
they hold  essentially for any chaotic quantum system,  depending  only
on whether the system is without or with time reversal symmetry. 
The nonrigorous  but remarkably  effective semiclassical orbit theory~\cite{Sieber_2001, 
PhysRevLett.93.014103,PhysRevLett.98.044103,MR805089} based upon Gutzwiller's trace formula~\cite{MR1077246} and many follow-up works 
verified  this conjecture   for quantizations of a large family  of classical chaotic systems, e.g.
for certain billiards.

\begin{figure}
  \centering
  \begin{tikzpicture}
    \begin{axis}[width=15cm,height=9cm,xmode=log,ymode=log,no markers,xtick={1,22.36,pi*250},
      xticklabels={$1$,$\sqrt{N}$,$\pi N/2$}]
       \addplot[col1,very thick] table[col sep=comma,x index=0,y index=1] {GUESFF.csv};
       \addplot[col4,very thick] table[col sep=comma,x index=0,y index=2] {GUESFF.csv};
       \addplot[col2,very thick] table[col sep=comma,x index=0,y index=3] {GUESFF.csv};
       \legend{$\abs{\braket{e^{\ii t H}}}^2$,$\E\abs{\braket{e^{\ii t H}}}^2$, $\Std\abs{\braket{e^{\ii t H}}}^2$,$\sqrt{N}$,$N$}
    \end{axis}
 \end{tikzpicture}
  \caption{A typical \emph{slope-dip-ramp-plateau} picture for the spectral form factor  of a chaotic system.
  The figure on log-log scale shows the SFF of a single GUE realisation \(H\) of size \(500\times 500\), as well as the empirical mean and standard deviation obtained from \(500\) independent realisations. }\label{fig sff}
\end{figure}
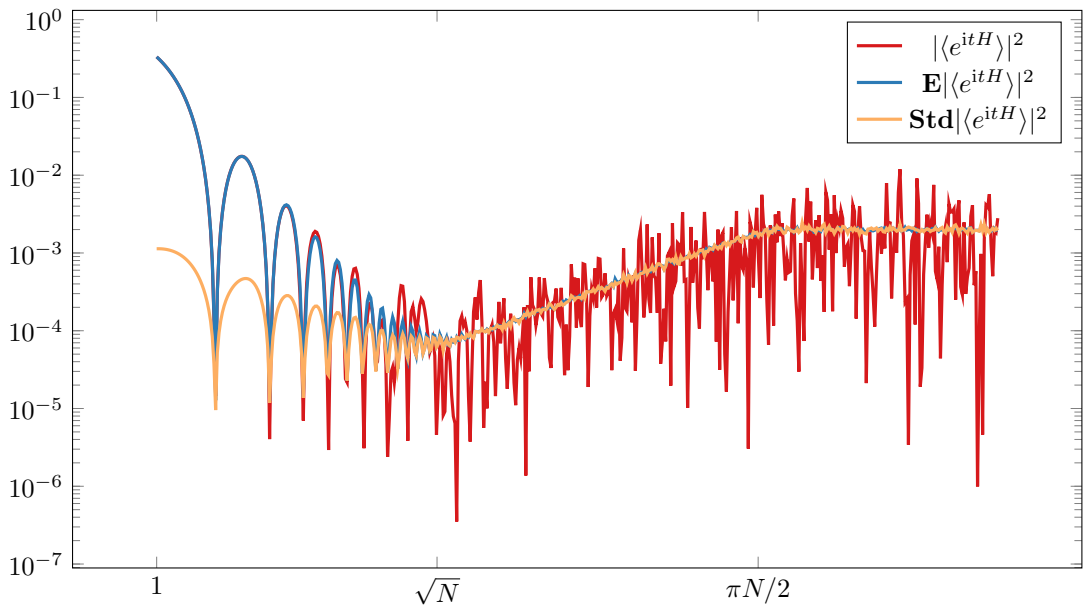

 For smaller times, $t\ll T_H$,  other details of $H$ may become
relevant. In particular the drop   from $K(t=0)=1$ to  $K(t)\ll 1$ for $1\ll t\ll T_H$  
is first dominated by the typical  non-analyticity of the density of states at the spectral edges
giving rise to the \emph{slope regime} up to 
an intermediate minimum point of $K(t)$, called the \emph{dip} (in the early literature the
dip was called \emph{correlation hole}~\cite{10032995},
for a recent overview, see~\cite{1706.05400}). 

\Cref{fig sff} shows the typical \emph{slope-dip-ramp-plateau} picture for the GUE ensemble.
Formula~\eqref{KGUEt} is valid starting from scales   $t\gg N^{1/2}$, while $K(t)$  is
oscillatorily decreasing for $t\lesssim N^{1/2}$  with a dip-time $t_{\mathrm{dip}}\sim N^{1/2}$.
Thus $K(t)$ follows
 the universal behavior~\eqref{KGUEt} only for $t\gg  t_{\mathrm{dip}}$. In this regime the fluctuation
 of the SFF is comparable with its expectation, $K(t)$, in fact  $\braket{e^{itH}}$ is approximately Gaussian.
 In contrast,  the dominant contribution to the slope regime, $t\ll t_{\mathrm{dip}}$,  is self-averaging
 with a relatively  negligible fluctuation.  However, if the edge effects are properly discounted (e.g. by considering
 the circular ensemble with uniform spectral density on the unit circle), i.e.\ 
  the slope regime is entirely removed, then the Gaussian behavior holds 
  for all $t\ll T_H$ with a universal variance given by~\eqref{KGUEt}.

In more recent works spectral form factors were studied for the celebrated Sachdev-Ye-Kitaev (SYK)
 model \cite{1806.06840,1611.04650, MR4136970, PhysRevD.96.066012, PhysRevD.97.106003} which also exhibits a similar \emph{slope-dip-ramp-plateau} pattern 
 although the details are still debated in the physics literature and the numerics are much less reliable 
 due to the exponentially  large dimensionality of the model.

\subsection{Our results}

Quite surprisingly, despite its central role in the physics literature on quantum chaos,
SFF has not been rigorously  investigated in the mathematics literature up to very recently,
when Forrester computed the large $N$ limit of $K(t)$ rigorously for the GUE in~\cite{MR4257834} and 
the Laguerre Unitary Ensemble (LUE) in~\cite{MR4312363} in the entire regime $t\ll N$. Both results
rely on a remarkable identity from~\cite{cond-mat/9608116} (and its extension to the LUE case)
and on previous 
stimulating work of Okuyama~\cite{MR3925257}. %
However, these methods use exact identities and thus are restricted to a few
explicitly solvable invariant ensembles.  

The main goal of the current paper is to investigate SFF beyond these special cases
with a robust method, the multi-resolvent local laws. 
While our approach is valid for quite
general ensembles, for definiteness we focus on two models: the standard Wigner ensemble
(for both symmetry classes) and the novel \emph{monoparametric ensemble} introduced 
recently~\cite{GPSW} by Gharibyan, Pattison, Shenker and Wells.
 The latter consists of matrices of the form $H^s:= s_1 H_1+s_2H_2$,
where $H_1$  and $H_2$ are typical but \emph{fixed} realisations of two independent Wigner  matrices
and $s=(s_1, s_2)\in S^1\subset \R$ is a continuous  random variable.
The normalization $s_1^2+s_2^2=1$ guarantees that the semicircle law for $H^s$ is independent of $s$
and it also shows that the model has effectively only one random parameter.
One may also consider similar ensembles with finitely  many  parameters (see Remark~\ref{rmk:m})
resulting in qualitatively the same behavior but with different power laws, see \Cref{table exps}. 

 We  study the statistics of $H^s$ in the
probability space of the single random variable $s$ and probe how much  universality still persists with 
such reduced randomness.
We write $\E_s$ for the expectation wrt. $s$ and $\E_H$, $\Std_H$ for the expectation
and standard deviation wrt. $H_1$ and $H_2$.

 Our  main result is to prove a formula for the expectation and standard deviation of SFF for both  ensembles
up to an intermediate time.  While this does
not include the  ramp regime, it already allows us to draw the following two main conclusions
of the paper:
\begin{enumerate}[label=(\alph*)]
\item The expectation and standard deviation of $\SFF(t)$ for Wigner   and monoparametric
ensembles exhibit the same universal  behavior to leading order for $1\ll t\ll N^{1/4}$ if the trivial edge effects
are removed. In the monoparametric case it is quite
remarkable that already a single real random variable generates universality. 
\item For the monoparametric ensemble  $K(t)=\E_s [\SFF(t)]$ depends non-trivially on the fixed $H_1, H_2$ matrices, 
but for large $t$ this dependence is a subleading effect whose relative size  becomes
increasingly negligible as a negative power of $t$. In particular, while the speed of convergence to universality is much slower
for the monoparametric ensemble than for the Wigner case,  it is improving for larger $t$.
 \end{enumerate}
The second item  answers a question raised by the authors of~\cite{GPSW} 
 which strongly motivated
 the current work. In particular, sampling from $s$ does not give a consistent estimator 
 for  $K(t)$, but the relative precision of such estimate improves for larger times.

We supplement these proofs with an extensive numerics demonstrating that both
conclusions hold  not only for $t\ll N^{1/4}$ but for the entire ramp regime, i.e.\ up to $t\ll T_H\sim N$.
Note that recently we have proved~\cite{2106.10200}
that the Wigner-Dyson-Mehta eigenvalue gap universality holds for the monoparametric ensemble,
which strongly supports, albeit does not prove,  that  $K(t)$ in the  plateau regime is also universal.

We remark that our method applies without difficulty  for finite temperatures (expressed by 
 a parameter $\beta>0$)  and for different-time autocorrelation functions, 
 i.e.\ for 
 \[
 \braket{e^{(-\beta+\ii t) H}}\braket{  e^{(-\beta-\ii t')H}}
 \]
 as well, but for the simplicity of the presentation we focus on $\SFF(t)$ defined in~\eqref{def:sff1},
 i.e.\ on $\beta=0$ and $t=t'$.

\subsection{Relations to previous mathematical results}

 Rigorous mathematics for the spectral form factor, even for Wigner matrices or even for GOE,
 significantly lags behind  establishing the compelling physics picture about the slope-dip-ramp-plateau.
 Given the recently developed tools in random matrix theory, it may appear 
 surprising that they do not directly answer  the important  questions on SFF. We now briefly explain why.

\subsubsection{Limitations of  the resolvent methods}

For problems on macroscopic spectral scales
 (involving the cumulative effect of  order $N$ many eigenvalues), and to a large extent also on
 mesoscopic scales (involving many more than $O(1)$ eigenvalues), the \emph{resolvent method} is suitable. This method
 considers the resolvent $G(z)= (H-z)^{-1}$ of $H$ for a spectral parameter $z$
 away from (but typically still close to)  the real axis and establishes that in a
 certain sense $G(z)$ becomes deterministic. This works for $\eta=\Im z\gg N^{-1}$ (in the bulk spectrum),
 i.e.\ on scales just above the eigenvalue spacing (note that the imaginary part of the 
 spectral parameter sets a scale in the spectrum). Such results are called \emph{local laws} and they can 
 be extended to regular functions $f(H)$ by standard spectral calculus (Helffer-Sj\"ostrand formula, see~\eqref{eq:HS} later).

 However, the interesting questions about   SFF  concern a $1/N$ subleading fluctuation
 effect beyond the local laws.   Indeed
 $$
   \Tr e^{\ii tH} = \sum_i e^{\ii t \lambda_i}
 $$
 is a special case of the well-studied \emph{linear eigenvalue statistics},  $\Tr f(H)=\sum_i f(\lambda_i)$,
 with the regular test function $f(\lambda)= e^{\ii t\lambda}$. 
To leading order it is deterministic  and
 its fluctuation satisfies the central limit theorem (CLT) without the customary $\sqrt{N}$  normalisation, i.e.
  \begin{equation}\label{linstat}
     \sum_i f(\lambda_i) - \E  \sum_i f(\lambda_i)\approx \mathcal{N}\big(0, V_f), 
    \qquad \mbox{with}\quad \E  \sum_i f(\lambda_i) = N \int_\R f(x) \rho_\mathrm{sc}(x)\dif x + O_f(1).
 \end{equation}
 is a normal random variable with variance\footnote{Eq.~\cref{Vf} is for matrices whose second and fourth moments coincide with the ones of GUE, otherwise there are additional terms, see e.g.~\cite[Theorem 2.4]{2012.13218}.}
 \begin{equation}\label{Vf}
 V_f = \frac{1}{4\pi^2}\iint_{-2}^2 \abs*{\frac{f(x)-f(y)}{x-y}}^2\frac{4-xy}{\sqrt{4-x^2}\sqrt{4-y^2}}\dif x \dif y .
\end{equation} 
The computation of higher moments of $\Tr  f(H) - \E \Tr  f(H)$ requires a generalization of 
the local laws to polynomial combinations of several $G$'s that are called \emph{multi-resolvent local laws}.

 Applying \eqref{linstat}--\eqref{Vf} to $f(x)= e^{\ii tx}$ we obtain, roughly, 
  \begin{equation}\label{Ot}
  \SFF(t) = \frac{1}{N^2}\big| \Tr e^{\ii t H}\big|^2
   \approx \Big[\frac{J_1(2t)}{t}  + O\Big( \sqrt{\frac{V_f}{N^2} } \Big)  \Big]^2, \qquad t\gg 1,  %
  \end{equation}
  using that 
 $$
\int_\R f(x) \rho_\mathrm{sc}(x)\dif x = \int_\R e^{\ii t x} \rho_\mathrm{sc}(x)\dif x = \frac{J_1(2t)}{t},
$$
where $J_1$ is the first Bessel function of the first kind.  Note that $V_f$ in~\eqref{Vf} scales essentially as the $H^{1/2}$ Sobolev norm of $f$
 hence $V_f\sim t$ for our $f(x)= e^{\ii tx}$ in the regime $t\gg1$.
Therefore the size of the  fluctuation term  in~\eqref{Ot} is $V_f/N^2\sim t/N^2$ and it competes
with the deterministic term $(J_1/t)^2\sim t^{-3}$. The dip time $t_{\mathrm{dip}}\sim\sqrt{N}$ is obtained
as the threshold where the fluctuation (the linear ramp function) becomes bigger than the slope function $(J_1/t)^2$.
This argument, however, is heuristic as it neglects the error terms in~\eqref{linstat} that also depend on $t$ via $f$.

CLT for linear statistics~\eqref{linstat} for Wigner matrices $H$
 has been proven~\cite{MR1487983, MR1411619, MR2189081, MR2829615, MR3116567,MR1899457,MR4095015, MR3678478,   2012.13218, 2105.01178, MR2489497,MR4187127,2103.05402, MR3959983} %
   for  test functions of the form 
$ f(x)= g(N^a(x-E))$
with some fixed reference point $|E|<2$, scaling exponent $a\in [0, 1)$ and  smooth function $g$
with compact support, i.e for macroscopic ($a=0$) and mesoscopic ($0<a<1$) test functions living on a \emph{single} scale $N^{-a}$.
These  proofs give optimal error terms for such functions but they
were not optimized
for dealing with  functions that oscillate on a mesoscopic scale \emph{and} have macroscopic support, 
like $f(x)= e^{\ii t x}$ for some $t\sim N^\alpha$, $\alpha>0$.
The only CLT-type result for a  special two-scale observable is in~\cite{MR4522353} where  the 
eigenvalue counting function smoothed on an intermediate scale $N^{-1/3}$  was considered.

Quite remarkably, extensive numerics shows that the formulas \eqref{linstat}--\eqref{Vf} for 
$f(x)= e^{\ii t x}$ are in perfect agreement with the expected behavior of 
 $K(t)$ in the entire slope-dip-ramp regime 
all the way up to $t\ll N$, i.e.\ the CLT for linear statistics 
correctly predicts SFF well beyond its proven regime of validity.
 In the current paper we optimise the error terms specifically for $e^{\ii tx}$
and thus we could
  cover the regime $t\ll N^{5/11}$ for the variance  in~\eqref{linstat} (corresponding to $\E [ \SFF(t)]$).

\subsubsection{Limitations of Dyson Brownian motion techniques}
For the microscopic scale (i.e.\ comparable with the eigenvalue spacing, $1/N$ in the bulk)
the resolvent is heavily fluctuating as it strongly depends on single eigenvalues. Local laws
cannot access them,  but in this regime another approach, the 
careful analysis of the \emph{Dyson Brownian Motion (DBM)} becomes applicable.
While these two approaches are  complementary and apparently  cover 
all scales, the actual methods require  additional conditions that seriously  restrict their use
for SFF. 

The 
formulas~\eqref{KGUE} are obtained
by  computing
the Fourier transform of  the  two point correlation function of the rescaled \emph{(unfolded)} eigenvalues. Indeed, in the GUE case $K_\mathrm{GUE}(t)$ in~\eqref{KGUE} is just the Fourier transform of $p_2(x,y) - 1+\delta(x-y)$
 in the difference variable $x-y$, where 
\[ 
p_2(x,y) := 1- \Big(\frac{\sin(\pi(x-y))}{\pi(x-y)} \Big)^2,
\]
is  the two point function, given by the celebrated Wigner-Dyson sine kernel,
and $K_\mathrm{GOE}(t)$ has a similar origin. 
Wigner-Dyson theory is designed for microscopic scales, i.e.\ to describe
eigenvalue correlations on scales comparable with the local level spacing $\Delta$,
this is encoded in the fact that~\eqref{KGUE} holds for any fixed $\tau>0$ in the $N\to \infty$ limit
(equivalently that~\eqref{KGUEt} holds only for $t\ge \delta N$ since $\Delta\sim 1/N$ in the bulk).
While this is a very elegant argument supporting~\eqref{KGUE},
mathematically it is quite far from a rigorous proof.

The mathematical proofs of the sine-kernel universality use test functions that are
rapidly decaying beyond scale $\Delta$. The typical statements (so called \emph{fixed energy universality}~\cite{MR3541852, 1609.09011}) show that for any fixed energy $E$ in the bulk 
\[
    \sum_{i<j} g\big( N\rho_\mathrm{sc}(E)(\lambda_i-E), N\rho_\mathrm{sc} (E)(\lambda_j-E) \big) \to \iint_\R
    g(x,y) p_2(x,y)\dif x\dif y
\]
in the large $N$ limit, for any smooth, compactly supported functions $g\colon\R^2\to \R$.
The current methods for proving the Wigner-Dyson universality 
 cannot deal with functions that are macroscopically supported, like
$g(x,y)= e^{\ii t(x-y)}$ with a fast oscillation $t\sim N$. 

\subsection{Summary}  
Using multi-resolvent local laws we prove a CLT for linear statistics of monoparametric ensembles (Theorem~\ref{CLT f(H)}) with covariance
\[
\Cov ( \Tr f(H^s) , \Tr g(H^r) ) \approx \frac{1}{\pi^2}\iint f'(x) g'(y) \log \abs*{\frac{ 1-\braket{s,r}m_\mathrm{sc}(x)\overline{m_\mathrm{sc}(y)}}{1- \braket{s,r}m_\mathrm{sc}(x)m_\mathrm{sc}(y)}} \dif x \dif y
\]
with an additional term depending on the fourth cumulant. Due to a careful analysis of the error terms this allows us to prove the expected behavior on the expectation and standard deviation
of the SFF for Wigner matrices for $t\ll N^{5/17}$  
(Theorem~\ref{sff wigner}) and for the monoparametric ensemble for $t\ll N^{1/4}$ (Theorem~\ref{sff mono}).  Beyond these regime 
the spectral form factor  is not understood mathematically apart from the special GUE and LUE cases.
However, we can still use our predictions from the CLT for linear statistics~\eqref{linstat} to derive an Ansatz for the behavior of $\SFF(t)$ in the entire $t\ll N$ regime. In particular, we show that the SFF is universal  for the  monoparametric ensemble. We find numerically that  our theory correctly reproduces $\SFF(t)$ for any $t\ll N$ and it also coincides with the physics predictions for the GUE case.

\subsection*{Notations and conventions}

For positive quantities $f$ and $g$ we will frequently use the notation  $f\approx  g$   meaning that $f/g \to 1$
in a limit that is always  clear from the context. Similarly, $f\ll g$ means that $f/g\to 0$. Finally, the relation $f\sim g$ means
 that there exist two positive constants $c, C$  such that $c\le f/g \le C$.

We say that an event holds "with very high probability" if for any fixed $D>0$ the probability of the event is bigger than $1-N^{-D}$ if $N\ge N_0(D)$, for some $N_0(D)>0$.

\subsection*{Acknowledgement} We are
grateful to the authors of~\cite{GPSW} for sharing with us their insights and preliminary numerical results. We are especially
thankful to Stephen Shenker for very valuable advice over several email communications.
Helpful comments  on the manuscript from Peter Forrester and from the anonymous referees are also acknowledged. 

\section{Statement of the main results}
Our new results mainly concern the monoparametric ensemble but for comparison reasons we also 
prove the analogous results for the Wigner ensemble. We start  with the two corresponding  definitions.
\begin{definition}\label{def:wigner} The Wigner ensemble consists of Hermitian $N\times N$ 
random matrices $H$ with the following properties. The off-diagonal matrix elements below the diagonal are independent, identically distributed (i.i.d) real (\(\beta=1\)) or complex \((\beta=2)\) random variables; in the latter case we assume that \(\E h_{ij}^2=0\). The diagonal elements are i.i.d.\ real 
random variables with \(\E  h_{ii}^2 = 2/(N\beta)\). Besides the standard normalisation~\eqref{Hnorm}, we also make the customary moment assumption: for every  $q\in \N$ 
there is a constant $C_q$ such that
\begin{equation}\label{mom}
    \E \big| \sqrt{N} h_{ij}\big|^{q}\le C_q.
\end{equation}
In the case of Gaussian distributions, it  is called the Gaussian Orthogonal  or Unitary Ensemble (GOE/GUE), for the real and complex cases, respectively.
\end{definition}
\begin{remark}
  The assumptions \(\E h_{ij}^2=0\) in the complex case, and \(\E h_{ii}^2=2/(\beta N)\) are
  made purely for convenience. All results can easily be generalised beyond this case but we refrain from doing so for notational simplicity.
\end{remark}
\begin{definition}\label{def:mon} The monoparametric ensemble consists of Hermitian $N\times N$ random matrices of the form 
\begin{equation}\label{Hs def}
  H= H^s:=s_1H_1+s_2H_2,
 \end{equation}
where $H_1, H_2$ are independent Wigner matrices satisfying\footnote{We assume equal fourth cumulants merely for notational convenience. Our proof verbatim covers also the more general case.} \(\E\abs{ h_{ij}^{(1)}}^4=\E\abs{ h_{ij}^{(2)}}^4\) and $s=(s_1,s_2)\in S^1$ is a random vector, independent of $H_1, H_2$. On the distribution of $s$ we assume that it has an square integrable
 density $\rho(s)$ independent of \(N\). 
 We write $\E_s$ for the expectation wrt.\ $s$ and $\E_H$, $\Std_H$ for the expectation and standard deviation wrt. the Wigner matrices $H_1$ and $H_2$.
\end{definition}

The parameter space $S^1\subset\R^2$ inherits the usual scalar product and norms from $\R^2$, so for $s,r\in S^1$ we have
\[
\braket{s,r}:=s_1r_1+s_2r_2,\quad \norm{s}_p:=(\abs{s_1}^p+\abs{s_2}^p)^{1/p}.
\]
We also introduce the entrywise product of two vectors:
\[
s\odot r:=(s_1r_1,s_2r_2).
\]
For a fixed $s$, $H^s$ is just the weighted sum of two Wigner matrices, and, due to the normalisation, itself is just a Wigner matrix.
However, the concept of monoparametric ensemble  views $H^s$  as a random matrix
in the probability space of the single random variable $s$ for a typical but fixed (quenched) realization of $H_1$
and $H_2$. 
While Wigner matrices have a large $(\sim N^2)$ number of independent random degrees of 
  freedom, the monoparametric ensemble is generated by one single random variable
  hence, naively, much less universality properties are expected. Nevertheless, the standard
  Wigner-Dyson local eigenvalue universality holds~\cite{2106.10200}.

\begin{remark}\label{rmk:m}
    In~\cite{2106.10200} we considered the un-normalized monoparametric model \(H^s:=H_1+sH_2\), for some real valued random variable $s$, whose density of states is a rescaled semicircular distribution.
   In this paper we prefer to work with more homogeneous models since the formulas are somewhat nicer, but 
    our main results   
    also apply to this inhomogeous model  with some slightly different exponents in the error terms.
    One may also consider a different  un-normalized ensemble, \(s_1H_1+s_2H_2\) with \(s\in\R^2\)
    having an absolutely continuous distribution, which is effectively a two parameter model.
    Similar results also hold for the multi-parametric analogue of~\cref{Hs def}, i.e.\ \(s_1H_1+\cdots+s_kH_k\) for \(s\in S^{k-1}\), see Remark~\ref{rmk multi} and Section~\ref{sec:ext} later.  Despite all these options, for definiteness, the main body of this paper concerns
   the homogenous monoparametric model from Definition~\ref{def:mon}.
\end{remark}

\subsection{Central limit theorem for sum of Wigner matrices}
To understand the effect of the random $s$, we  study
the joint statistics of $H^s$ and $H^r$ 
for two different fixed realisations $r, s$ in the probability space of $H_1, H_2$,
i.e.\ we aim at the correlation effects between $H^s$ and $H^r$. We introduce the short-hand notations 
\begin{equation}
  \braket{f}_\mathrm{sc}:=\int_{-2}^2 f(x)\frac{\sqrt{4-x^2}}{2\pi}\dif x, \quad \braket{f}_\mathrm{1/sc}:=\int_{-2}^2 f(x)\frac{1}{\pi\sqrt{4-x^2}}\dif x, \quad \kappa_4 := N^2 \E\abs{h_{12}}^4 -1-\frac{2}{\beta}.
\end{equation}
To estimate the error term in the following theorem we introduce a parameter $1\le \tau\ll N$ and the weighted norm
  \begin{equation}
 \label{eq:taunorm}
  \norm{f}_\tau:=\tau^2\norm{f}_\infty+\tau\norm{f}_{H^1}+\norm{f}_{H^2},
  \end{equation}
  where \(\norm{f}_{H^k}^2 :=\sum_{j\le k}\int_\R\abs{f^{(j)}}^2\) is the usual Sobolev norm.
   For the applications later, the parameter \(\tau\) will be optimized.

\begin{theorem}\label{CLT f(H)}For \(s\in S^1\) and test functions \(f\in H^2(\R)\) the family of random variables \(\Tr f(H^s)\) is approximately Gaussian of mean
  \begin{equation}\label{exp fH}
    \E \Tr f(H^s) = N\braket{f}_\mathrm{sc} + \kappa_4\norm{s}_4^4 \braket*{\frac{x^4-4x^2+2}{2}f}_\mathrm{1/sc} + \bm1(\beta=1)\Bigl[\frac{f(2)+f(-2)}{4}  -\frac{\braket{f}_\mathrm{1/sc}}{2}\Bigr] + \landauO{\cE_1},
  \end{equation}
  and fluctuation
  \begin{equation}
  \label{eq:vsrwick}
    \E \prod_{i=1}^p \Bigl(\Tr f_i(H^{s^i})-\E \Tr f_i(H^{s^i})\Bigr) = \sum_{P\in\mathrm{Pair}([p])} \prod_{(i,j)\in P} v^{s^is^j}(f_i,f_j) + \mathcal O_p(\cE_p),
  \end{equation}
 for any fixed $p\in \mathbf{N}$, functions $f_1,\dots, f_p\in H^2(\R)$, and parameters $s^1,\dots, s^p\in S^1$, 
 where\footnote{For the applications in this paper, SFF in the regime $t\gg 1$,  
 the first term in~\eqref{vsrfg} is the only relevant one.}
  \begin{align}  \label{vsrfg}
      v^{sr}(f,g) &:= \frac{1}{\beta \pi^2} \iint_{-2}^2 f'(x)g'(y) V^{sr}(x,y)\dif x \dif y + \frac{\kappa_4}{2} \braket{s\odot s,r \odot r}\braket{(2-x^2)f}_\mathrm{1/sc}^2 \\\nonumber
      V^{sr}(x,y)&:=  \log \abs*{1-\braket{s,r}m_\mathrm{sc}(x)\ov{m_\mathrm{sc}(y})} - \log \abs*{1-\braket{s,r} m_\mathrm{sc}(x)m_\mathrm{sc}(y)}.
  \end{align}
Here \(\cE_p\) are error terms which for any \(1\le\tau\ll N\) and any $\xi,\epsilon>0$ may be estimated by\footnote{The exponent in~\eqref{eq:errorssff} can be optimized depending on $\tau$ and \(f\).}
  \begin{equation}
  \label{eq:errorssff}
  \cE_1:=\frac{N^\xi \norm{f}_\tau}{N^{1/2}\tau^{1/2}}, \quad   \cE_p :=N^\xi \left(\frac{1}{N^{1/2}\tau^{3/2}}+\frac{N^\epsilon}{N}+\frac{N^{-\epsilon}}{\tau^{2p-1}}\right)\left(1+\frac{\tau^2}{N^{1-2\epsilon}}\right)\prod_{i\in [p]}\norm{f_i}_\tau,
  \end{equation} 
for $p\ge 2$. Additionally, if $s^1=\dots=s^p$, i.e.\ in the Wigner case, we have the improved bound
  \begin{equation}
  \cE_p := \frac{1}{N^{1/2}\tau^{3/2}}\prod_{i\in [p]}\norm{f_i}_\tau
  \end{equation}
  and the first term of~\cref{vsrfg} for \(\beta=2\) coincides with~\cref{Vf}. 
\end{theorem}
We note that~\eqref{vsrfg} generalizes the standard variance calculation yielding~\eqref{Vf} to \(s\ne r\),
see Section~\ref{s=r}.

\begin{remark}\label{rmk multi}
  \Cref{CLT f(H)} verbatim holds true also for the multi-parametric model
  \[s_1H_1+\cdots+s_kH_k\]
  upon interpreting \(\braket{s,r}\) and \(\norm{s}_p\) as the Euclidean inner product and \(p\)-norm in \(\R^k\). Similarly,~\Cref{CLT f(H)} also applies to the un-normalised case \(s\in\R^2\) for which on the rhs.\ of~\cref{exp fH} the function \(f\) has to be
   replaced by \(f(\norm{s}\cdot)\) with $\norm{\cdot}:=\norm{\cdot}_2$  and \(v^{sr}\) from~\cref{vsrfg} 
   has to be replaced by  
  \begin{equation}\label{vsrfg tilde}
    \begin{split}
      \wt v^{sr}(f,g)&:= \frac{\norm{s}\norm{r}}{\beta\pi^2} \iint_{-2}^2 f'(\norm{s}x)g'(\norm{r}y) 
      V^{\frac{s}{\| s\|},\frac{r}{\|r\|}}(x,y)\dif x \dif y \\
      &\qquad + \frac{\kappa_4}{2} \braket{s\odot s,r \odot r}\braket{(2-x^2)f(\norm{s}x)}_\mathrm{1/sc}\braket{(2-x^2)f(\norm{r}x)}_\mathrm{1/sc}.
    \end{split}
  \end{equation}
\end{remark}

\subsection{SFF for Wigner and monoparametric ensemble}

In this section we
specialise~\Cref{CLT f(H)} to the SFF  case. We define the approximate  expectation (rescaled by $1/N$)
  \begin{equation}
    \begin{split}
      e_N^s(t)&:= e(t) + \frac{1}{N}\Bigl[\kappa_4\norm{s}_4^4\left(1-\frac{6}{t^2}\right) J_0(2 t)+\kappa_4\norm{s}_4^4\left(\frac{6}{t^{3}}-\frac{4}{t}\right) J_1(2 t) - \bm1(\beta=1)\frac{J_0(2 t) - \cos(2 t)}{2}\Bigr] \\
      e(t)&:=\frac{J_1(2t)}{t}
    \end{split}
  \end{equation}
  in terms of the Bessel functions \(J_k\) of the first kind.    We also define the approximate variance
  \begin{equation}\label{vpm kappa}
    \begin{split}
      v_{\pm,\kappa}^{sr}(t)&:=v^{sr}(e^{\ii t \cdot},e^{\pm\ii t \cdot})=v_\pm^{sr}(t)+\kappa_4 \braket{s\odot s,r\odot r} J_2(2t)^2,\\
      v_\pm^{sr}(t) &:=   \frac{t^2}{\beta\pi^2} \iint_{-2}^2 \cos\Bigl(t (x\pm y)\Bigr)V^{sr}(x,y)\dif x \dif y,
    \end{split}
  \end{equation}

From Theorem~\ref{CLT f(H)}, choosing $f_i(x)=e^{\pm \ii tx}$ and $\tau=t$, and recalling that $\braket{e^{\pm \ii t H^s}}=N^{-1}\mathrm{Tr}\, e^{\pm \ii t H^s}$, we readily conclude the following asymptotics for SFF of the Wigner and monoparametric ensemble. 
\begin{figure}
  \centering
  \begin{tikzpicture}
    \begin{axis}[width=15cm,height=7.5cm,xmode=log,ymode=log,no markers,legend columns=2]
       \addplot [line width=3pt, domain=11:120, samples=10,col1!20] {2/pi*x/10000};
       \addplot [line width=3pt, domain=200:1000, samples=10,col1!20,dashed] {1/100};
       \addplot[col1] table[col sep=comma,x index=0,y index=1] {ffEGUE.csv};
       \addplot[dotted,very thick,col1] table[col sep=comma,x index=0,y index=1] {EGUE.csv};
       \addplot[col4] table[col sep=comma,x index=0,y index=2] {ffEGUE.csv};
       \addplot[dotted,very thick,col4] table[col sep=comma,x index=0,y index=2] {EGUE.csv};
       \legend{$2t/(\pi N^2)$,$1/N$,$\E_H \abs{\braket{e^{\ii t H}}}^2$,$E_\mathrm{wig}(t)$,$\Std_H\abs{\braket{e^{\ii t H}}}^2$,$S_\mathrm{wig}(t)$}
    \end{axis}
 \end{tikzpicture} 
    \begin{tikzpicture}
     \begin{axis}[width=15cm,height=7.5cm,xmode=log,ymode=log,no markers,legend columns=2]
        \addplot[col1] table[col sep=comma,x index=0,y index=1] {ffEstd.csv};
        \addplot[dotted,very thick,col1] table[col sep=comma,x index=0,y index=1] {Eint.csv};
        \addplot[thick,col4] table[col sep=comma,x index=0,y index=3] {ffEstd.csv};
        \addplot[dotted,very thick,col4] table[col sep=comma,x index=0,y index=2] {Stdint.csv};
        \addplot [line width=3pt, domain=15:120, samples=10,col2!30] {.8*x^(.75)/10000};
        \addplot [line width=3pt, domain=200:1000, samples=10,col2!30,dashed] {1.5*x^(-.5)/100^0.75};
        \addplot[thick,col2] table[col sep=comma,x index=0,y index=2] {ffEstd.csv};
        \addplot[dotted,very thick,col2] table[col sep=comma,x index=0,y index=1] {Stdint.csv};
        \legend{$\E_H \E_s \abs{\braket{e^{\ii t H^s}}}^2$,$E_\mathrm{wig}(t)$,$\bigl(\E_H \Var_s\abs{\braket{e^{\ii t H^s}}}^2\bigr)^{1/2}$,$\bigl(S_\mathrm{wig}(t)^2-S_\mathrm{res}(t)^2\bigr)^{1/2}$,$0.8t^{3/4}/N^2$,$1.5t^{-1/2}/N^{3/4}$,$\Std_H \E_s \abs{\braket{e^{\ii t H^s}}}^2$,$S_\mathrm{res}(t)$} 
     \end{axis}
  \end{tikzpicture}
  \caption{In the first plot we compare the empirical mean (red) and standard deviation (blue) of \(\abs{\braket{e^{\ii t H}}}^2\) obtained from sampling \(10,000\) independent \(100\times 100\) GUE matrices \(H\) with our approximation~\cref{GUE SFFbis}.
  In the second plot we similarly compare the empirical mean (red) and 
  variance (blue), with respect to $s$, 
   obtained from sampling \(500\) independent scalar random variables \(s\) (from the uniform distribution on $S^1$) 
  and \(500\) independent \(100\times100\) GUE matrix pairs \(H_1,H_2\), with the prediction~\cref{eq FF concl}. 
  We also test the precision of the latter GUE-pair sampling by finding the empirical standard deviation (with respect to \(H_1,H_2\))
  of the empirical mean of the monoparametric SFF (orange).
  We observe that for both ensembles our resolvent approximation seems valid for  all \(t<N\).}  \label{fig std}
\end{figure}
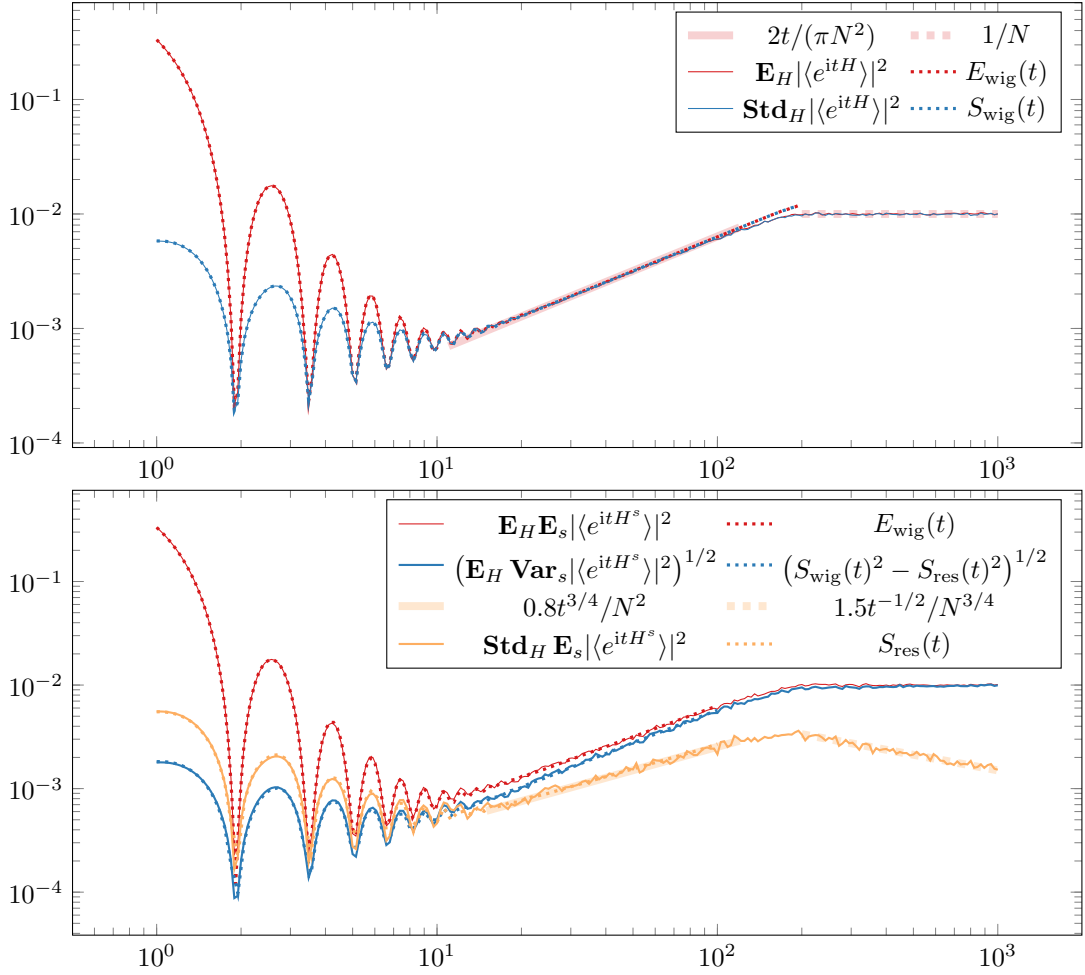

\begin{theorem}[SFF for the Wigner ensemble]\label{sff wigner}
  For deterministic \(t>0\) (possibly $N$-dependent) we have 
  \begin{equation}
  \label{GUE SFFbis}
    \begin{split}
      \E_H \abs{\braket{e^{\ii t H}}}^2 &= E_\mathrm{wig}(t)(1+o(1))  \quad \mathrm{for} \quad t\ll N^{5/11},\\
      \Var_H \abs{\braket{e^{\ii t H}}}^2 &=S_\mathrm{wig}(t)^2(1+o(1)) \quad \mathrm{for} \quad t\ll N^{5/17},
    \end{split}
  \end{equation}
  and we have the asymptotics 
  \begin{equation}\label{asymp}
    \begin{split}
      E_\mathrm{wig}(t):=e(t)^2 + \frac{v^{ee}_{-,\kappa}(t)}{N^2} &\approx \begin{cases}   \frac{J_1(2t)^2}{t^2}, &  1\ll t\ll \sqrt{N} \\
         \frac{2}{\pi} \frac{t}{N^2}, &
          \sqrt{N}\ll t\ll N,\end{cases} \\
S_\mathrm{wig}(t):=\biggl(\frac{v_{+,\kappa}^{ee}(t)^2+v_{-,\kappa}^{ee}(t)^2}{N^4} +  2e(t)^2\frac{v_{+,\kappa}^{ee}(t)+v_{-,\kappa}^{ee}(t)}{N^2}\biggr)^{1/2} &\approx \begin{cases}  \frac{2J_1(2t)}{\sqrt{\pi t}N}  , & 1\ll  t\ll \sqrt{N} \\
       \frac{2}{\pi} \frac{t}{N^2}, &
         \sqrt{N}\ll t\ll N,\end{cases}
    \end{split}
    \end{equation}
    where we set \(e:=(1,0)\in S^1\).
    \end{theorem}
  This result shows that
  $ S_\mathrm{wig}(t)\ll E_\mathrm{wig}(t)$ in the slope regime, $t\ll \sqrt{N}$,
   and $ S_\mathrm{wig}(t)\approx E_\mathrm{wig}(t)$
    in the ramp regime, $\sqrt{N}\ll t\ll N$ (see the first plot in Figure~\ref{fig std}). In particular, in the ramp regime the SFF is a non-negative random variable whose fluctuations are of the same size as its expectation.  Thus the SFF  is not self-averaging in 
    the ramp regime, while it is self-averaging in the slope regime but only 
    owing to the dominance  of the function $e(t)$ representing the edge effect.
  If one discounts the edge effect, i.e.\ artificially  removes $e(t)$, then 
  $ S_\mathrm{wig}(t)\approx E_\mathrm{wig}(t)$  would hold for all $1\ll t\ll N$,
  demonstrating the universal behavior of SFF in the entire slope-dip-ramp regime.

\begin{theorem}[SFF for the monoparametric ensemble~\eqref{Hs def}]\label{sff mono}
  For \(t>0\) (possibly $N$-dependent) we have 
    \begin{equation}
  \label{eq FF concl}
    \begin{split}
        \E_H\E_s\abs{\braket{e^{\ii t H^s}}}^2 &= E_\mathrm{wig}(t)(1+o(1))\quad \mathrm{for} \quad t\ll N^{3/7}  \\
        \E_H\Var_s\abs{\braket{e^{\ii t H^s}}}^2 & = \Bigl(S_\mathrm{wig}(t)^2-S_\mathrm{res}(t)^2\Bigr)(1+o(1))\quad \mathrm{for}\quad t\ll N^{5/17}\\
        \Var_H\E_s\abs{\braket{e^{\ii t H^s}}}^2 & = S_\mathrm{res}(t)^2 (1+o(1))\quad \mathrm{for} \quad t\ll N^{1/4}
    \end{split}
    \end{equation}
    where the function
    \begin{equation}\label{ES}
      \begin{split}
        S_\mathrm{res}(t)&:=\sqrt{\E_s\E_r \Bigl(\frac{v_{+,\kappa}^{sr}(t)^2+v_{-,\kappa}^{sr}(t)^2}{N^4} +  2e(t)^2\frac{v_{+,\kappa}^{sr}(t)+v_{-,\kappa}^{sr}(t)}{N^2}\Bigr)}
      \end{split}
    \end{equation}
    satisfies 
        \begin{equation}\label{EHmonSTmon}
\begin{split}
S_\mathrm{res}(t)  \sim \begin{cases}    \frac{\psi(t)}{Nt^{5/4}}  , &  1\ll t\ll \sqrt{N} \\
        \frac{t^{3/4}}{N^2}, &
         \sqrt{N}\ll t\ll N,\end{cases}
\end{split}
   \end{equation}
   where $\psi(t)\sim 1$ is a positive function with some oscillation.
   \end{theorem}
   In particular, this result immediately  shows the following concentration effect:
\begin{corollary}
  For \(1\ll t\ll N^{1/4}\) it holds that
  \begin{equation}\label{cor Var bd}
    \Var_H \E_s\abs{\braket{e^{\ii t H^s}}}^2 \lesssim \frac{1}{\sqrt{t}} \Var_H \abs{\braket{e^{\ii t H}}}^2,
  \end{equation}
  i.e.\ averaging in \(s\) reduces the size of the fluctuation of the SFF by a factor of $t^{-1/4}$. 
\end{corollary}

 Note that  
 \be\label{compa}
 S_\mathrm{res}(t) \lesssim t^{-1/4} S_\mathrm{wig}(t)
 \ee
 both in the slope and ramp regimes showing that not only the 
 expectation but also the variance of the SFF for the monoparametric
 ensemble coincide with those for the Wigner ensemble to leading order,
  hence they follow the  universal pattern (red and blue curves in the second plot in Figure~\ref{fig std}). 
   However, the dependence 
 of $\E_s [\SFF(t)]$ on the fixed Wigner matrix pair $(H_1, H_2)$ is 
 still present, albeit to a lower order, expressed by the \emph{residual standard deviation} $S_\mathrm{res}(t)$ whose relative size  decreases as $t^{-1/4}$
 as $t$ increases (orange curves in  Figure~\ref{fig std}). It is quite remarkable that a single random mode $s$ generates almost the entire randomness in the ensemble that is responsible for the  universality of SFF. A similar phenomenon was manifested in the Wigner-Dyson universality  
 proven in~\cite{2106.10200}.

\begin{remark}\label{ext}
 Based upon extensive numerics (see \Cref{fig std}) we believe that \eqref{GUE SFFbis}, \eqref{eq FF concl} and~\eqref{cor Var bd}
 hold up to any $t\ll N$, i.e.\ in the entire slope-dip-ramp regime and not only up
 to some fractional power of $N$ as stated and  proved rigorously. The proof for the entire regime $t\ll N$
 is out of reach with the current technology based upon
the multi-resolvent local law Lemma~\ref{23g ll} whose error term does not trace the expected improvement due to different spectral parameters $z_1\ne z_2$. We expect that the entire ramp regime \(t\ll N\) should be accessible by resolvent techniques if a sharp version of Lemma~\ref{23g ll}, tracing the gain from \(z_1\ne z_2\), was available.
 \end{remark}
 
\begin{remark}
  We stated~\Cref{sff wigner,sff mono} only for the first two moments but the CLT from~\Cref{CLT f(H)} allows us to compute arbitrary moments $\E\abs{\braket{e^{\ii t H}}}^{2m}$ for the Wigner case and $\E_s\abs{\braket{e^{\ii t H^s}}}^{2m}$ for the monoparametric case (together with their concentration in the $(H_1, H_2)$-space), albeit with worsening error estimates. This would lead  to rigorous results of the type~\eqref{GUE SFFbis} and~\eqref{eq FF concl} but for a shorter time scale $t\ll N^{c(m)}$ with some $c(m)>0$. However,
  in the spirit
  of Remark~\ref{ext},  we believe that \(\braket{e^{itH^s}}\)
   can be approximated for  any \(t\ll N\), to leading order,  by a family of complex 
  Gaussians \(\xi(t,s)\) of mean and variance
  \begin{equation}\label{xi cov}
    \E \xi(t,s) = e(t), \quad \E (\xi(t,s)-e(t))(\xi(t',s')-e(t')) = \frac{1}{N^2} v^{ss'}(e^{\ii t\cdot},e^{\ii t' \cdot})
  \end{equation}
  with \(v^{sr}\) from~\eqref{vsrfg}. Note that~\cref{xi cov} also specifies the covariance of \(\xi(t,s)\) and \(\ov{\xi(t',s')}=\xi(-t',s')\)
  for different times.
\end{remark}

The next lemma, to be proven in Section~\ref{s=r}, provides explicit asymptotic formulas for $v^{ss}_\pm(t)$, in particular they imply
the asymptotics in~\eqref{asymp} together with $e(t)\sim t^{-3/2}$ (up to some oscillation due to the Bessel function)
 in the large $t$  regime.

\begin{lemma}\label{lemma vss}
  For \(s=r\) the functions  $ v^{ss}_\pm(t)$ appearing in~\eqref{vpm kappa} can be expressed as
  \begin{equation}\label{v00pm}
  \begin{split}
     v^{ss}_-(t) &=t^2 \Bigl[J_0(2t)^2 + 2J_1( 2 t)^2 - 
     J_0( 2 t)J_2( 2 t)\Bigr] = \frac{2 t}{\pi }-\frac{1+2\sin(4t)}{16 \pi  t} + \landauO{t^{-2}}\\
     v^{ss}_+(t) &=- t J_0(2t) J_1(2t) =\frac{\cos (4 t)}{2 \pi }-\frac{2+\sin (4 t)}{16 \pi  t}+ \landauO{t^{-2}}.
  \end{split}
  \end{equation}  
\end{lemma}

The relation in~\eqref{EHmonSTmon} requires a stationary phase calculation
that will be done separately in Section~\ref{sec:statphase}.

\subsection{Implications for sampling}
\begin{figure}
  \centering
    \begin{tikzpicture}
     \begin{axis}[width=15cm,height=7cm,xmode=log,ymode=log,no markers,legend columns=3,legend style={/tikz/column 4/.style={column sep=5pt,},}]
        \addlegendimage{no markers,black};
        \addlegendentry{$\E_H^n \abs{\braket{e^{\ii t H}}}^2$};
        \addlegendimage{no markers,line width=8pt,black!20};
        \addlegendentry{$E_\mathrm{wig}(t) \pm n^{-1/2}S_\mathrm{wig}(t)$};
        \addlegendimage{no markers,white};
        \addlegendentry{};
        \addplot[col2,thick,forget plot] table[col sep=comma,x index=0,y index=1] {ffGUEstrip.csv};
        \addplot[col4,thick,forget plot] table[col sep=comma,x index=0,y index=2] {ffGUEstrip.csv};
        \addplot[col1,thick,forget plot] table[col sep=comma,x index=0,y index=3] {ffGUEstrip.csv};
        \addplot[col2,name path=A2,forget plot] table[col sep=comma,x index=0,y index=1] {GUEstrip.csv};
        \addplot[col2,name path=B2,forget plot] table[col sep=comma,x index=0,y index=2] {GUEstrip.csv};
        \addplot[col2!40] fill between[of=A2 and B2];
        \addplot[col4,name path=A10,forget plot] table[col sep=comma,x index=0,y index=3] {GUEstrip.csv};
        \addplot[col4,name path=B10,forget plot] table[col sep=comma,x index=0,y index=4] {GUEstrip.csv};
        \addplot[col4!40] fill between[of=A10 and B10];
        \addplot[col1,name path=A500,forget plot] table[col sep=comma,x index=0,y index=5] {GUEstrip.csv};
        \addplot[col1,name path=B500,forget plot] table[col sep=comma,x index=0,y index=6] {GUEstrip.csv};
        \addplot[col1!40] fill between[of=A500 and B500];
        \addlegendentry{$n=2$};
        \addlegendentry{$n=10$};
        \addlegendentry{$n=500$}; %
     \end{axis}
  \end{tikzpicture}
  \begin{tikzpicture}
    \begin{axis}[width=15cm,height=7cm,xmode=log,ymode=log,no markers,legend columns=3,legend style={/tikz/column 4/.style={column sep=5pt,},}]
       \addlegendimage{no markers,black};
       \addlegendentry{$\E_s^n \abs{\braket{e^{\ii t H^s}}}^2$};
       \addlegendimage{no markers,line width=8pt,black!20};
       \addlegendentry{$E_\mathrm{wig}(t)\pm \max(n^{-1/2} S_\mathrm{wig}(t),S_\mathrm{res}(t))$};
       \addlegendimage{no markers,white};
       \addlegendentry{};
       \addplot[col2,thick,forget plot] table[col sep=comma,x index=0,y index=1] {ffMonostrip.csv};
       \addplot[col4,thick,forget plot] table[col sep=comma,x index=0,y index=2] {ffMonostrip.csv};
       \addplot[col1,thick,forget plot] table[col sep=comma,x index=0,y index=3] {ffMonostrip.csv};
       \addplot[col2,name path=A2,forget plot] table[col sep=comma,x index=0,y index=1] {Monostrip.csv};
       \addplot[col2,name path=B2,forget plot] table[col sep=comma,x index=0,y index=2] {Monostrip.csv};
       \addplot[col2!40] fill between[of=A2 and B2];
       \addplot[col4,name path=A10,forget plot] table[col sep=comma,x index=0,y index=3] {Monostrip.csv};
       \addplot[col4,name path=B10,forget plot] table[col sep=comma,x index=0,y index=4] {Monostrip.csv};
       \addplot[col4!40] fill between[of=A10 and B10];
       \addplot[col1,name path=A500,forget plot] table[col sep=comma,x index=0,y index=5] {Monostrip.csv};
       \addplot[col1,name path=B500,forget plot] table[col sep=comma,x index=0,y index=6] {Monostrip.csv};
       \addplot[col1!40] fill between[of=A500 and B500];
       \addlegendentry{$n=5$};
       \addlegendentry{$n=20$};
       \addlegendentry{$n=1000$}; %
    \end{axis}
 \end{tikzpicture}
  \caption{In the first plot we show the empirical mean of \(\abs{\braket{e^{\ii tH}}}^2\) for \(k\) independent GUE matrices \(H\). As expected the standard deviation of the sample average fluctuates within a strip of width \(n^{-1/2}\Std_H\abs{\braket{e^{\ii t H}}}^2\),
  in particular the sample average  exactly reproduces the mean if $n\to\infty$.  
   In the second plot we show the empirical mean of \(\abs{\braket{e^{\ii tH^s}}}^2\) for \(k\) independently sampled
    scalar random variables \(s\) for a fixed GUE matrix pair \(H_1,H_2\). We observe that while the sample mean approximates the true mean \(\E_s\) increasingly well as \(n\to\infty\), the latter is still dependent on the chosen realisation of \(H_1,H_2\). Thus the empirical mean fluctuates in a strip of width \(\max(n^{-1/2}S_\mathrm{wig}(t),S_\mathrm{res}(t))\) around the doubly averaged \(\E_H\E_s\abs{\braket{e^{\ii tH^s}}}^2\).}\label{fig strip}
\end{figure}
Determining the standard deviation of $\abs{\braket{e^{\ii tH}}}^2$ is important
for numerical testing of~\eqref{GUE SFFbis}.  By taking the empirical average \(\E_H^n\) of \(n\) independent Wigner matrices we may approximate the true expectation \(\E_H\abs{\braket{e^{\ii tH}}}^2\) at a speed 
\begin{equation}\label{eq strip GUE}
  \E_H^n \abs{\braket{e^{\ii tH}}}^2 = \E_H \abs{\braket{e^{\ii tH}}}^2 + \Omega\Big( n^{-1/2}\Std_H \abs{\braket{e^{\ii tH}}}^2\Big)=E_\mathrm{wig}(t)+\Omega(n^{-1/2}S_\mathrm{wig}(t)),
\end{equation}
c.f.\ the top of Figure~\ref{fig strip}.
Here  $\Omega(\cdots)$ indicates an oscillatory error term of the given  size.
 In the \emph{ramp regime} the fluctuation of 
 \(\E_H^n \abs{\braket{e^{\ii tH}}}^2\) thus scales like \(t/(\sqrt{n}N^{2})\)
 using~\eqref{asymp}. In particular, this fluctuation vanishes as the sample size $n$ goes to infinite,
 hence the statistics via sampling to test~\eqref{GUE SFFbis} is  \emph{consistent}. 

In contrast, for the monoparametric ensemble, by taking the empirical average of \(n\) copies of \(s\) we naturally have 
\begin{equation}
  \E_s^n \abs{\braket{e^{\ii t H^s}}}^2 = \E_s \abs{\braket{e^{\ii t H^s}}}^2 + \Omega\Big( k^{-1/2} S_\mathrm{wig}(t)\Big).
\end{equation}
Replacing the first term by its expectation plus its fluctuation in the $H$-probability space, we also get
\begin{equation}\label{EHEs}
    \E_s^n \abs{\braket{e^{\ii t H^s}}}^2
  = \E_H\E_s \abs{\braket{e^{\ii t H^s}}}^2 + \Omega\Bigl(\max\bigl(n^{-1/2}S_\mathrm{wig}(t),S_\mathrm{res}(t)\bigr)\Bigr),
\end{equation}
where the error term contains both standard deviations and satisfies
\begin{equation}\label{eq sampl}
  \max\bigl(n^{-1/2}S_\mathrm{wig}(t),S_\mathrm{res}(t)\bigr) 
  \sim \begin{cases} \frac{1}{Nt}\max\{ \frac{1}{\sqrt{n}},\frac{1}{t^{1/4}}\} , & 1\ll  t\ll \sqrt{N} \\
   \frac{t}{N^2}\max\{ \frac{1}{\sqrt{n}},\frac{1}{t^{1/4}}\} &
    \sqrt{N}\ll t\ll N,\end{cases}
\end{equation}
due to~\eqref{eq FF concl} and~\eqref{EHmonSTmon}.
In particular,  both in the slope and 
 in the ramp regimes %
  the size of the fluctuation of \(\E_s^n\abs{\braket{e^{\ii t H^s}}}^2\) 
 does not vanish even as the number of samples goes to infinity, $n\to \infty$,
hence the statistics is \emph{not consistent}, c.f.\ the bottom of~\Cref{fig strip}. 
However, this lack of consistency, expressed by $S_\mathrm{res}(t)$ is still negligible
compared with the leading first term in~\eqref{EHEs} by a factor $t^{-1/4}\ll 1$
in the large $t$ regime, see~\eqref{compa}.
We recall that mathematically rigorously we can prove all these facts
 only for $t\ll N^{1/4}$, i.e.\ well before the dip time, but the numerical tests leave no doubt on 
 their validity in the entire regime $1\ll t\ll N$.

\subsection{Extensions}\label{sec:ext}
Beside the Wigner ensemble, we formulated our main results on SFF for the normalized monoparametric model
in~\Cref{sff mono}. We chose this model for definiteness, but our approach applies to the multi-parametric
as well as to the un-normalised models introduced in Remark~\ref{rmk:m}. Here we explain the modified results for
these natural generalisations.

First, for  the multi-parametric normalised model, $H^s=s_1H_1+\ldots + s_kH_k$ with $k-1$ effective parameters
$s\in S^{k-1}$,
~\Cref{sff mono} holds true verbatim modulo  different sizes for the residual standard deviation \(S_\mathrm{res}(t)\).
In fact,  we have  
\be\label{Sk}
S_\mathrm{res}(t) \lesssim t^{-\frac{1}{2}+\frac{1}{4}(3-k)_+} S_\mathrm{wig}(t),
\ee
 see~\eqref{tar4} later, 
hence $S_\mathrm{res}(t) $ becomes less relevant compared with $S_\mathrm{wig}(t)$
for larger $k>2$, see~\eqref{compa}. Consequently, the upper bounds on the  times of proven validity in~\cref{eq FF concl} slightly 
improve but they still remain below the dip time and we omit the precise formulas. We note that the \(t\)-power in~\cref{Sk} is not optimal for \(k\ge 3\). A refined stationary phase estimate could be used to improve the estimate but we refrain from doing so since our primary interest is the mono-parametric model with few degrees of freedom.

Second, for the un-normalised 
  model $H^s=s_1H_1+ s_2H_2$ with two effective parameters \(s\in \R^2\), ~\Cref{sff mono} also holds true modulo 
  some minor changes. More precisely,~\cref{eq FF concl} becomes
  \begin{equation}
    \label{eq FF concl change}
      \begin{split}
          \E_H\E_s\abs{\braket{e^{\ii t H^s}}}^2 &= \E_s E_\mathrm{wig}(\norm{s}_2t)(1+o(1))\quad \mathrm{for} \quad t\ll N^{3/7}  \\
          \E_H\Var_s\abs{\braket{e^{\ii t H^s}}}^2 & = \Bigl(\E_s S_\mathrm{wig}(\norm{s}_2t)^2-\wt S_\mathrm{res}(t)^2\Bigr)(1+o(1))\quad \mathrm{for}\quad t\ll N^{5/17}\\
          \Var_H\E_s\abs{\braket{e^{\ii t H^s}}}^2 & = \wt S_\mathrm{res}(t)^2 (1+o(1))\quad \mathrm{for} \quad t\ll N^{1/7},
  \end{split}
  \end{equation}
  with \(\wt S_\mathrm{res}\) obtained from replacing \(v^{sr}\) by \(\wt v^{sr}\) from~\Cref{rmk multi} in~\cref{vpm kappa}. 
  For \(\wt S_\mathrm{res}(t)\) a  stationary phase calculation gives the modified 
  \begin{equation}\label{tildeS}
    \wt S_\mathrm{res}(t)  \sim  \begin{cases}    \frac{\psi(t)}{Nt^{7/4}}  , &  1\ll t\ll \sqrt{N} \\
      \frac{t^{1/4}}{N^2}, &
       \sqrt{N}\ll t\ll N, \end{cases}
  \end{equation}
  assuming that $s$ has an absolutely  continuous distribution with a differentiable, compactly supported
   density $\rho$ on $\R^2$ with $\rho(0)=0$.
  We will not prove the relation in these formulas in this paper, we only  show how to obtain
  the necessary upper bound on  them at
  the end of Section~\ref{sec:statphase}.
  
  Note that now 
  \be\label{Sres2}
  \wt S_\mathrm{res}(t)\lesssim t^{-3/4}\E_s S_\mathrm{wig}(\norm{s}_2t),
  \ee
  i.e.\ the fluctuation due to  the residual randomness of $(H_1, H_2)$  after taking 
  the expectation in $s$ remains negligible, in fact it is reduced compared with the normalised case~\eqref{compa}.
  As a consequence \(t^{1/4}\) in~\cref{eq sampl} is replaced by \(t^{3/4}\).
  
 Analogous results  hold for the most general multi-parametric un-normalised model
 as well as to the mono-parametric inhomogeneous model $H^s=H_1+sH_2$, $s\in \R$.
 We omit their precise formulation, the key point is that the analogue of~\eqref{eq FF concl change}
 hold in all cases with a residual standard deviation $\wt S_\mathrm{res}(t)$ being smaller 
 than the leading term $S_\mathrm{wig}(t)$ by a polynomial factor in $t$  (e.g. by $t^{-1/2}$
 for $H^s=H_1+sH_2$).
  This guarantees that 
 the universality of SFF holds for  all these models. \Cref{table exps} summarizes the decay exponents of our main parametric models.

 \begin{table}[ht]
 \centering
 \caption{For our three main parametric models the following table lists the size of the residual fluctuation compared to the fluctuation of the Wigner-SFF.}\label{table exps}
 \begin{tabular}[t]{ccc}
 \toprule
 Quenched parametric model&Randomness&\\
 \midrule
 $s_1 H_1 + s_2 H_2$&$(s_1, s_2)\in S^1$&$S_\mathrm{res}(t) \lesssim t^{-1/4} S_\mathrm{wig}(t)$\\
 $H_1 + s H_2$&$s\in\R$&$S_\mathrm{res}(t) \lesssim t^{-1/2} S_\mathrm{wig}(t)$\\
 $s_1 H_1 + s_2 H_2$ &$(s_1, s_2)\in \R^2$&$S_\mathrm{res}(t) \lesssim t^{-3/4} S_\mathrm{wig}(t)$\\
 \bottomrule
 \end{tabular}
 \end{table}%

\subsection*{Outline} The rest of the paper is organised as follows. In \Cref{sec:linstat} we outline the resolvent method and explain how via the Helffer-Sj\"ostrand representation a resolvent CLT implies the CLT for the linear statistics \(\sum f(\lambda_i)\) of arbitrary test functions \(f\) from which our main results Theorems~\ref{CLT f(H)}--\ref{sff mono} follow. In \Cref{sec:CLTres} we present the proof of the resolvent CLT, while in Section~\ref{sec:statphase} 
we conclude the proof of the asympotics \cref{EHmonSTmon} via a stationary phase argument.

\section{Resolvent method}\label{sec:linstat}

Let $H$ be a Wigner matrix\footnote{The resolvent method extends to very general Hermitian 
matrices possibly with  non-centered and correlated entries, see~\cite{MR3941370}, but here we present
only the Wigner case for simplicity.} and $G(z):= (H-z)^{-1}$ its resolvent with a spectral parameter $z\in \C\setminus \R$.
Define $m_\mathrm{sc}(z)$, the Stieltjes transform of the semicircle law:
\begin{equation}
\label{eq:semicirc}
   m(z)=m_\mathrm{sc}(z) :=\int_\R \frac{\rho_\mathrm{sc}(x)}{x-z} \dif x, \quad \rho_\mathrm{sc}(x) := \frac{\sqrt{(4-x^2)_+}}{2\pi}.
\end{equation}

The \emph{local law} for a single resolvent states that the diagonal matrix $m(z)\cdot I$
 well approximates the random resolvent $G(z)$ in the following sense (see e.g. \cite{MR3183577, MR2871147, MR3103909}):
\begin{equation}\label{loclaw}
   |\braket{(G(z)-m(z))A}| \lesssim N^\xi \frac{\|A\|}{N\eta}, \qquad  \langle {\bm x}, (G(z)-m(z)) {\bm y}\rangle\lesssim N^\xi \frac{\| {\bm x}\| \|{\bm y}\|}{\sqrt{N\eta}}
\end{equation}
with $\eta = |\Im z|$, for any fixed deterministic  matrix $A$ and deterministic vectors ${\bm x}, {\bm y}$.
The first bound is called \emph{averaged} local law, while the second one is the \emph{isotropic} local law. 
The bounds \eqref{loclaw} are understood in very high probability for any fixed $\xi>0$. 

The Helffer-Sj\"ostrand formula 
\begin{equation}
\label{eq:HS}
\braket{f(H)}=\frac{2}{\pi}\int_\mathbf{C}\partial_{\overline{z}}f_\mathbf{C}(z) \braket{G(z)}\dif^2 z,
\end{equation}
with $z=x+\ii\eta$ and $\dif^2 z:=\dif \eta\dif x $, expresses the linear statistics of arbitrary functions as an integral 
of the resolvent \(G(z)\) and the almost-analytic extension
\begin{equation}
\label{eq:aaeaa}
f_\mathbf{C}(z)=f_\mathbf{C}(x+\ii\eta):=\big[f(x)+\ii\eta\partial_x f(x)\big]\chi(\tau\eta),
\end{equation}
of \(f\). Here the free parameter $\tau\in\mathbf{R}$ is chosen such that $N^{-1}\ll \tau^{-1}\lesssim 1$, and $\chi$ a smooth cut-off equal to $1$ on $[-5,5]$ and equal to $0$ on $[-10,10]^c$. The same $\tau$ was used to define the weighted 
$H^2$-norm~\eqref{eq:taunorm} and eventually  we will optimize its value, a procedure  that improves
the standard error terms in the CLT.
By~\Cref{loclaw} it follows that 
\begin{equation}\label{lin stat lead}
  \braket{f(H)} = \frac{2}{\pi} \int_\C \partial_{\overline{z}}f_\mathbf{C}(z) m(z)\dif^2 z + \landauO*{N^\xi\frac{\norm{f}_{H^2}}{N}} = \int_{-2}^2 \rho_\mathrm{sc}(x) f(x)\dif x +  \landauO*{N^\xi\frac{\norm{f}_{H^2}}{N}}. 
\end{equation}

In order to compute the fluctuation in~\cref{lin stat lead} via~\cref{eq:HS} we need to understand the correlation between \(\braket{G(z)},\braket{G(z')}\) for two different spectral parameters \(z,z'\) which turns out to be given by 
\begin{equation}
  \Cov(\braket{G(z)},\braket{G(z')}) \approx \frac{1}{N^2}\frac{\braket{G(z)^2}\braket{G(z')^2}\braket{G(z)G(z')}(1+\braket{G(z)G(z')})}{\braket{G(z)}\braket{G(z')}},
\end{equation}
modulo some additional contribution from non-Gaussian fourth cumulant, see~\cref{eq:wickresmain} for the final statement. While $G(z)\approx m(z)$, in general it is not true that $G(z) G(z') \approx m(z)m(z')$ 
since \eqref{loclaw}  allows only deterministic test matrices  multiplying  $G$. Nevertheless $G(z) G(z')$ is still approximable by a deterministic object:
\begin{equation}\label{gg}
  G(z) G(z') \approx \frac{m(z) m(z')}{1- m(z) m(z')}.
\end{equation}
 Statements of the form~\eqref{gg} with an appropriate error term are called \emph{multi-resolvent local laws}.

We will apply this theory to the product of the resolvents $G^s$ of $H^s=s_1H_1+s_2H_2$ for two different 
parameters $s$, see the corresponding local law on $\braket{G^s G^r}$ in~\eqref{eq:23gl}
later. Even though  $H_1$ and $H_2$  as well as $s$  and  $r$ are independent, 
the common $(H_1, H_2)$ ingredients in $H^s$ and $H^r$ introduce a nontrivial
correlation between these matrices. We therefore need to extend CLT for resolvents via
   multi-resolvent local laws to 
this parametric situation.

\subsection{Resolvent CLT}
The main technical result of the present paper is the following Central Limit Theorem for product of resolvents of the random matrix $H^s:=s_1H_1+s_2H_2$ with $s=(s_1,s_2)\in S^1$.

\begin{proposition}
\label{pro:resCLT}
Fix $\epsilon>0$, $p\in \mathbf{N}$, $s^1,\dots,s^p\in S^1$, $z_1,\dots,z_p\in\mathbf{C}\setminus\mathbf{R}$, and define $G_i:=(H^{s^i}-z_i)^{-1}$. Then for any arbitrary small $\xi>0$ and $\eta_\ast\ge N^{-1+\epsilon}$ it holds
\begin{equation}
\label{eq:wickresmain}
\E_H\prod_{i\in [p]}\braket{G_i-\E_H G_i}=\frac{1}{N^p}\sum_{P\in \mathrm{Pair}([p])}\prod_{(i,j)\in P} V_{ij}+\mathcal{O}\left(N^\xi\Psi_p\left(\frac{1}{L^{1/2}}+\frac{1}{N\eta_*^2}+\frac{1}{N^2\eta_*^4}\right)\right), \quad \Psi_p:=\prod_{i\in [p]}\frac{1}{N|\eta_i|}.
\end{equation}
Here $\eta_i:=\Im z_i$, $\eta_\ast:=\min_i\abs{\eta_i}$, $L:=\min_i (N\eta_i\rho_i)$, and
\begin{equation}\label{eq:defvij}
V_{ij}:=-\frac{2}{\beta}\partial_{z_i}\partial_{z_j}\log\biggl(1-\braket{s^i,s^j} m_im_j\biggr)-\braket{s^i\odot s^i,s^j\odot s^j}\kappa_4(m_i^2)'(m_j^2)',
\end{equation}
where $m_i:=m_{\mathrm{sc}}(z_i)$, and $\kappa_4:=N^2\E|h_{12}|^4-1-2/\beta$. Additionally, for the expectation we have
\begin{equation}
\label{eq:expmpres}
\E_H\braket{G_i}=m_i+\frac{\kappa_4}{N}\norm{s}_4^4 m_i'm_i^3+\bm1(\beta=1)\frac{1}{N}\frac{m_im_i'}{1-m_i^2}+\mathcal{O}\left(N^\xi\frac{\sqrt{\rho_i}}{(N\eta_i)^{3/2}}\right),
\end{equation}
with $\rho_i:=\pi^{-1}|\Im m_i|$.
\end{proposition}

\begin{remark}
For Wigner matrices, i.e.\ for $s^1=\dots=s^p=(1,0)$, the error term in \eqref{eq:wickresmain} is given by $\Psi L^{-1/2}$, as a consequence of the fact that the error terms in the first and second line of \eqref{eq:23gl} are replaced by $(N\eta_1\eta_2)^{-1}$ and $(N\eta_1\eta_2^2)^{-1}$, respectively (see e.g. \cite[Remark 3.5]{MR4372147}).
\end{remark}

We point out that similar resolvent CLT have often been  used as a basic input to prove CLT for linear eigenvalue 
statistics of both Hermitian and non-Hermitian matrices down to optimal mesoscopic scales (see e.g. \cite{2210.12060, MR4420179, MR4551555, MR3959983, MR4095015, MR3678478, 2204.03419, MR4187127, 2105.01178}). The main novelty
here is to extend the resolvent CLT to the monoparametric ensemble.

Along the proof of \Cref{pro:resCLT} we establish the following multi-resolvent local laws.
\begin{lemma}\label{23g ll}
For \(G_i=G^{s^i}(z_i)\) we have the two- and three-resolvent local laws
  \begin{equation}
    \label{eq:23gl}
    \begin{split}
      \abs*{\braket{G_1G_2}-\frac{m_1m_2}{1-\braket{s^1,s^2}m_1m_2}}&\lesssim \frac{N^\xi}{N|\eta_1\eta_2|^{3/2}}\\
      \abs*{\braket{G_1G_2^2}-\frac{m_1m_2'}{(1-\braket{s^1,s^2}m_1m_2)^2}}&\lesssim \frac{N^\xi}{N|\eta_1||\eta_2|\eta_*^2}+\frac{1}{N^2|\eta_1\eta_2|^3},
    \end{split}
  \end{equation}
  where \(m_i=m_\mathrm{sc}(z_i)\), with very high probability for any fixed $\xi,\epsilon>0$ and $|\Im z_i|\ge N^{-1+\epsilon}$.
\end{lemma}
The proofs of~\Cref{pro:resCLT} and~\Cref{23g ll} will be presented in~\Cref{sec:CLTres}. In these proofs we will often 
use the standard \emph{cumulant expansion} (see \cite{MR1689027, MR3678478, MR1411619} in the random matrix context):
\begin{equation}
\label{eq:cumexpsff}
\E_H h_{ab}f(H)=\frac{1}{N}\E_H\partial_{ba} f(H)+\sum_{k=2}^R\sum_{q+q'=k} \frac{\kappa_{ab}^{q+1,q'}}{N^{(k+1)/2}}\E_H \partial_{ab}^q\partial_{ba}^{q'}f(H)+\Omega_R.
\end{equation}
Here $\partial_{ab}$ denotes the directional derivative $\partial_{h_{ab}}$, the first
 term in the rhs. represents the second order (Gaussian) contribution, while the sum in \eqref{eq:cumexpsff} represents
  the non-Gaussian contribution with $\kappa^{p,q}_{ab}$ denoting the joint cumulant of $p$ copies of $N^{1/2}h_{ab}$ and $q$ copies of $N^{1/2}\overline{h_{ab}}$. The cumulant expansion is typically truncated at a high  ($N$-independent) level $R$ with an
error term $\Omega_R$ that is negligible. To see this, note that  in our applications $f$ will be a product of resolvents
at spectral parameters  $z_i$ with $\eta_* = \min |\Im z_i|\gg 1/N$ hence derivatives of $f$ remain bounded 
with very high probability by the isotropic local law~\eqref{loclaw} thus the tail of the series~\eqref{eq:cumexpsff} decays as $N^{-(k+1)/2}$.

\subsection{Proof of Theorem~\ref{CLT f(H)}}
The proof of Theorem~\ref{CLT f(H)} is divided into three steps: (i) computation of the expectation, (ii) computation of the variance, (iii) proof of Wick Theorem. The expectation is computed in Section~\ref{sec:subexp}, while the Wick Theorem and the explicit computation of the variance are proven in Section~\ref{sec:subvarwick}.

\subsubsection{Expectation}
\label{sec:subexp}

Using the bound
\begin{equation}
\label{eq:bfc}
\big|\partial_{\overline{z}}f_\mathbf{C}\big|\lesssim \eta |f''|+\tau |\chi'|\big[|f|+\ii\eta |f'|\big],
\end{equation}
and $|\braket{G^s-m}|\lesssim N^\xi (N\eta)^{-1}$ by \eqref{loclaw}, with $m=m_{\mathrm{sc}}$, we conclude that
\begin{equation}
\label{eq:ba}
\E_H\braket{f(H^s)}=\int_\mathbf{R}\int_{|\eta|\ge \eta_0}\partial_{\overline{z}}f_\mathbf{C}(z)\E_H\braket{G^s(z)} \dif\eta\dif x+\mathcal{O}\left(\frac{N^\xi \eta_0\norm{f}_{H^2}}{N}+N^\xi\eta_0^2\norm{f}_{H^2}\right),
\end{equation}
for any $N^{-1}\ll \eta_0\ll \tau^{-1}$. Note that we chose $\eta_0\gg N^{-1}$ in order to use Proposition~\ref{pro:resCLT}. 

Plugging \eqref{eq:expmpres} into \eqref{eq:ba}, and using \eqref{eq:bfc} to estimate the error term, we get that
\begin{equation}
\label{eq:kappa4exp}
\begin{split}
\E_H\braket{f(H^s)}&=\int_\mathbf{R}\int_{|\eta|\ge \eta_0}\partial_{\overline{z}}f_\mathbf{C}(z)\Big[m+\frac{\kappa_4}{N}\norm{s}_4^4m'm^3+\bm1(\beta=1)\frac{1}{N}\frac{mm'}{1-m^2}\Big]\dif\eta\dif x\\
&\quad+\mathcal{O}\left(\frac{N^\xi\eta_0\norm{f}_{H^2}}{N}+N^\xi\eta_0^2\norm{f}_{H^2}+\frac{N^\xi\norm{f}_{H^2}}{N^{3/2}\tau^{1/2}}+\frac{N^\xi\tau^{3/2}\norm{f}_\infty}{N^{3/2}}+\frac{N^\xi\tau^{1/2}\norm{f}_{H^1}}{N^{3/2}}\right) \\
&=\int_\mathbf{R}\int_{|\eta|\ge \eta_0}\partial_{\overline{z}}f_\mathbf{C}(z)\Big[m+\frac{\kappa_4}{N}\norm{s}_4^4m'm^3+\bm1(\beta=1)\frac{1}{N}\frac{mm'}{1-m^2}\Big]\dif\eta\dif x+\mathcal{O}\left(\frac{N^\xi\norm{f}_{\tau}}{N^{3/2}\tau^{1/2}}\right),
\end{split}
\end{equation}
where to go to the last line we chose $\eta_0\sim N^{-1+\epsilon}$, for some very small $\epsilon>0$, and we used the norm $\norm{f}_{\tau}$ defined in \eqref{eq:taunorm}.

Adding back the regime $|\eta|< \eta_0$ at the price of a negligible error smaller than the one in \eqref{eq:kappa4exp}, by explicit computations (exactly as in~\cite[Section D.1]{2012.13218}) in the leading term of \eqref{eq:kappa4exp}, we conclude 
\begin{equation}
\begin{split}
  \E_H \braket{f(H^s)} &= \int_{-2}^2 \rho_\mathrm{sc}(x) f(x)\dif x + \frac{\kappa_4}{2N}\norm{s}_4^4 \int_{-2}^2 \frac{x^4-4x^2+2}{\pi\sqrt{4-x^2}}f(x)\dif x \\
  &\quad+\bm1(\beta=1)\left[\frac{f(2)+f(-2)}{4N}-\frac{1}{2\pi N}\int_{-2}^2 \frac{f(x)}{\sqrt{4-x^2}}\,\dif x\right]+\mathcal{O}\left(\frac{N^\xi\norm{f}_{\tau}}{N^{3/2}\tau^{1/2}}\right).
  \end{split}
\end{equation}

\subsubsection{Second moment and Wick theorem}
\label{sec:subvarwick}
Define
\begin{equation}
\label{eq:linstats}
L_N(f,s):=N[\braket{f(H^s)}-\E_H \braket{f(H^s)}],
\end{equation}
then in this section, using Proposition~\ref{pro:resCLT}, we compute the leading order term of $\E_H L_N(f_1,s^1)L_N(f_2,s^2)$. More precisely, by \eqref{eq:wickresmain} for $p=2$, and using \eqref{eq:bfc} to estimate the error term, it follows that
\begin{equation}
\label{eq:varfs}
\begin{split}
&\E_H L_N(f_1,s^1)L_N(f_2,s^2)\\
&=\iint_\mathbf{R}\iint_{|\eta_1|,|\eta_2|\ge \eta_0}\partial_{\overline{z_1}}f_\mathbf{C}(z_1)\partial_{\overline{z_2}}f_\mathbf{C}(z_2) V_{12} \\
&\quad+\mathcal{O}\Bigg(N^\xi\eta_0(\norm{f_1}_{H^2}\norm{f_2}_\infty+\norm{f_2}_{H^2}\norm{f_1}_\infty)+\frac{N^\xi\norm{f_1}_\tau\norm{f_2}_\tau}{N^{1/2}\tau^{3/2}}+\frac{\norm{f_1}_{H^2}\norm{f_2}_{H^2}}{N^{1-\xi}\eta_0\tau}\left(1+\frac{1}{N\eta_0^2}\right)\\\
&\quad\qquad+\frac{(\norm{f_1}_{H^2}(\tau^2\norm{f_2}_\infty+\tau\norm{f_2}_{H^1})+\norm{f_2}_{H^2}(\tau^2\norm{f_1}_\infty+\tau\norm{f_1}_{H^1}))}{N^{1-\xi}\eta_0\tau}\left(1+\frac{1}{N\eta_0^2}\right)\\
&\quad\qquad+\frac{(\tau^2\norm{f_1}_\infty+\tau\norm{f_1}_{H^1})(\tau^2\norm{f_2}_\infty+\tau\norm{f_2}_{H^1})}{N}\left(1+\frac{\tau^2}{N^{1-2\epsilon}}\right)\Bigg)\\
&=\iint_\mathbf{R}\iint_{|\eta_1|,|\eta_2|\ge N^{-\epsilon}\tau^{-1}}\partial_{\overline{z_1}}f_\mathbf{C}(z_1)\partial_{\overline{z_2}}f_\mathbf{C}(z_2) V_{12}\\
&\quad+\mathcal{O}\left(N^\xi \norm{f_1}_\tau\norm{f_2}_\tau\left(\frac{N^\epsilon}{N}+\frac{N^{-\epsilon}}{\tau^3}\right)\left(1+\frac{\tau^2}{N^{1-2\epsilon}}\right)\right),
\end{split}
\end{equation}
where to go to the last line we chose $\eta_0\sim N^{-\epsilon}\tau^{-1}$, for any $\epsilon>0$, and $V_{12}$ is defined in \eqref{eq:defvij}. From \eqref{eq:varfs}, adding back the regimes $|\eta_i|< N^{-\epsilon}\tau^{-1}$ at the price of an error smaller than the one in the last line of \eqref{eq:varfs}, we conclude \eqref{eq:vsrwick} for $p=2$ by explicit computation in deterministic term as in \cite[Section D.2]{2012.13218}.

We conclude this section with the computation of higher moments:
\begin{equation}
\label{eq:fwick}
\begin{split}
\E_H\prod_{i\in [p]}L_N(f_i,s^i)&=\sum_{P\in \mathrm{Pair([p])}}\prod_{(i,j)\in P}\iint_\mathbf{R}\iint_{|\eta_i|,|\eta_j|\ge N^{-\epsilon}}\partial_{\overline{z_i}}f_\mathbf{C}(z_i)\partial_{\overline{z_j}}f_\mathbf{C}(z_j) V_{ij}\\
&\quad+\mathcal{O}\left(\left(\frac{N^\xi}{N^{1/2}\tau^{3/2}}+\frac{N^\epsilon}{N}+\frac{N^{-\epsilon}}{\tau^{2p-1}}\right)\left(1+\frac{\tau^2}{N^{1-2\epsilon}}\right)\prod_{i\in [p]}\norm{f_i}_\tau\right),
\end{split}
\end{equation}
which concludes the proof of \eqref{eq:vsrwick} for any $p\in\mathbf{N}$, after  adding back the regimes $|\eta_i|< N^{-\epsilon}\tau^{-1}$ at the price of an error smaller than the one in the second line of \eqref{eq:fwick}.

\subsubsection{Proof of Theorems~\ref{sff wigner} and~\ref{sff mono}}

We just show how Theorem~\ref{sff wigner} follows by Theorem~\ref{CLT f(H)}; the proof of Theorems~\ref{sff mono} is completely analogous and so omitted. In particular, to make the presentation shorter we just show the details of the proof of the first equation in \eqref{GUE SFFbis}. Using Theorem~\ref{CLT f(H)} as an input, the proof of the second equation in \eqref{GUE SFFbis} follows exactly in the same way.

First of all we write
\begin{equation}
\label{eq:lastphop}
\E_H|\braket{e^{\ii t H}}|^2=\E_H\big|\braket{e^{\ii t H}}-\E_H\braket{e^{\ii t H}}\big|^2+ \big|\E_H\braket{e^{\ii t H}}\big|^2.
\end{equation}
Then, using \eqref{eq:vsrwick} with $p=2$, $f_1(x)=e^{\ii tx}$, $f_2(x)=e^{-\ii tx}$, and $\tau=t$ to compute the leading order of the first term in \eqref{eq:lastphop}, and \eqref{exp fH} with $f(x)=e^{itx}$ to compute the leading order of the second term in \eqref{eq:lastphop}, we conclude that
\begin{equation}
\label{eq:fff}
\E_H|\braket{e^{\ii t H}}|^2=E_{\mathrm{wig}}(t)+\mathcal{O}\left(\frac{1}{N^{3/2}}+\frac{t^{5/2}}{N^{5/2}}\right),
\end{equation}
with $E_{\mathrm{wig}}(t)$ defined in \eqref{asymp}. Finally, using the asymptotics of $E_{\mathrm{wig}}(t)$ in \eqref{asymp} we readily conclude that the error term in \eqref{eq:fff} is much smaller than the leading term $E_{\mathrm{wig}}(t)$ as long as $t\ll N^{5/11}$.

\subsubsection{Variance calculations when \(s=r\) and the proof of  \Cref{lemma vss}}\label{s=r}
We note that~\eqref{vsrfg} generalises the standard variance calculation yielding~\eqref{Vf} to \(s\ne r\). For the case \(s=r\) the two formulas can be seen to be equivalent using the identity
\begin{equation}\label{id} \begin{split}
  \frac{1}{2\pi^2}\iint_{-2}^2 f'(x)g'(y) & \log\abs*{\frac{1-m_\mathrm{sc}(x)\ov{m_\mathrm{sc}(y)}}{1-m_\mathrm{sc}(x)m_\mathrm{sc}(y)}}\dif x\dif y  \\  &= \frac{1}{4\pi^2} \iint_{-2}^2 \frac{f(x)-f(y)}{x-y}\frac{g(x)-g(y)}{x-y} \frac{4-xy}{\sqrt{4-x^2}\sqrt{4-y^2}}\dif x \dif y
\end{split}
\end{equation}
that can be proven by integration by parts and using $(m_\mathrm{sc}(x) + x)m_\mathrm{sc}(x)=-1$
 from the explicit form of \(m_\mathrm{sc}(x)\) from~\eqref{eq:semicirc}.
\begin{proof}[Proof of \Cref{lemma vss}]
  Using~\eqref{id} the functions  $ v^{ss}_\pm(t)$ appearing in~\eqref{vpm kappa} can be expressed as 
  \begin{equation}\label{v00}
  \begin{split}
     v^{ss}_-(t) &=\frac{1}{\pi^2}\int_{-1}^1\int_{-1}^1
     \frac{1-x y}{ \sqrt{1-x^2} \sqrt{1-y^2} }\Bigl(\frac{\sin\left(t(x-y)\right)}{x-y}\Bigr)^2 \dif x\dif y\\
     &= \sum_{k=1}^\infty k J_{k}(2t)^2 = t^2 \Bigl[J_0(2t)^2 + 2J_1( 2 t)^2 - 
     J_0( 2 t)J_2( 2 t)\Bigr] \\
     &= \frac{2 t}{\pi }-\frac{1+2\sin(4t)}{16 \pi  t} + \landauO{t^{-2}}
  \end{split}
   \end{equation}
   and
   \begin{equation}\label{v00plus}
   \begin{split}
     v^{ss}_+(t) &=\frac{1}{4\pi^2}\int_{-1}^1\int_{-1}^1
     \frac{1-x y}{ \sqrt{1-x^2} \sqrt{1-y^2} }\Bigl(\frac{e^{2\ii tx} - e^{2\ii ty}}{x-y}\Bigr)^2 \dif x\dif y\\
     &=\sum_{k=1}^\infty (-1)^k k J_k(2t)^2 = - t J_0(2t) J_1(2t)\\
     &=\frac{\cos (4 t)}{2 \pi }-\frac{2+\sin (4 t)}{16 \pi  t}+ \landauO{t^{-2}}
   \end{split}
   \end{equation}
   where the series representations follow directly from~\cite[Remark 2.6]{2012.13218} and the series evaluations follow from~\cite[V.\S 5.51(1)]{MR1349110}.  
\end{proof}

\section{Central Limit Theorem for resolvents}
\label{sec:CLTres}
The proof of Proposition~\ref{pro:resCLT} is divided into three parts: in Section~\ref{sec:expres} we compute the subleading order correction to $\E_H \braket{G_i}$, in Section~\ref{sec:varres} we explicitly compute the variance, and finally in Section~\ref{sec:secwick} we prove a Wick Theorem. To keep our presentation simpler we only prove the CLT for resolvent in the complex case, the real case is completely analogous and so omitted (see e.g. \cite[Section 4]{2012.13218}).

\subsection{Computation of the expectation}
\label{sec:expres}
For \(G=G^s(z)\) we have 
\begin{equation}
\label{Geqa}
  I =s_1\un{H_1 G} + s_2 \un{H_2 G} -\braket{G}G - zG, \quad \un{H_i G}:=H_1 G+s_i\braket{G}G
\end{equation}
so that \(G\approx m\) for the solution \(m\) to the equation
\begin{equation}
\label{meqa}
  -\frac{1}{m} = z + m, \qquad m(z)=m_\mathrm{sc}(z).
\end{equation}
The fact that \(G\approx m\) in averaged and isotropic sense follows from the single resolvent local law \eqref{loclaw}. This is a consequence of the fact that the term $\un{H_i G}$ in \eqref{Geqa} is designed in such a way $\E \un{H_i G}\approx 0$ in averaged and isotropic sense. In fact, for Gaussian ensembles $\E \un{H_i G}=0$ and the deviation from zero for general ensembles
 is a lower order effect due to
non-vanishing of higher order cumulants of the entry distribution. 
From~\cref{Geqa,meqa} we obtain 
\begin{equation}
\label{eq:singgeq}
  (1-m^2\braket{\cdot})[G-m] = - m(s_1\un{H_1 G}+s_2\un{H_2G}) + m \braket{G-m}(G-m).
\end{equation}
Additionally, we define $\rho(z):=\pi^{-1}|\Im m(z)|$. For simplicity of notation from now on we assume that $\Im z>0$. We remark that by $1-m^2\braket{\cdot}$ in the lhs. of \eqref{eq:singgeq} we denote the operator acting on matrices $R\in\C^{N\times N}$ as $(1-m^2\braket{\cdot})[R]=R-m^2\braket{R}$.

We then start computing: 
\begin{equation}
\label{eq:geqsff}
\E_H \braket{G-m}=-\frac{m'}{m}\E_H\braket{s_1\un{H_1 G}+s_2\un{H_2G}} +\mathcal{O}\left(\frac{N^\xi}{N^2\eta^2\rho}\right),
\end{equation}
for any small $\xi>0$, where we used that $|1-m^2|\gtrsim \rho$, that \(m'=m^2/(1-m^2)\), and that $|\braket{G-m}|\lesssim N^\xi(N\eta)^{-1}$ by \eqref{loclaw}. Then using cumulant expansion (see \eqref{eq:cumexpsff}, ignoring the truncation error) we claim (and prove below) that
\begin{equation}
\label{eq:expa}
\begin{split}
\E_H\braket{s_1^1\un{H_1 G}+s_2^1\un{H_2G}}&=\E_H\frac{1}{N}\sum_{k\ge 2}\sum_{ab}\sum_{{\bm 
\alpha}\in\{ab,ba\}^k}\left(\frac{\kappa^{(1)}(ab,{\bm\alpha})}{k!}s_1\partial_{\bm\alpha}^{(1)}+\frac{\kappa^{(2)}(ab,{\bm\alpha})}{k!}s_2\partial_{\bm\alpha}^{(2)}\right) G_{ba} \\
&=\frac{\kappa_4}{N}\norm{s}_4^4 m^4+\mathcal{O}\left(\frac{N^\xi\rho^{3/2}}{N^{3/2}\eta^{1/2}}+\frac{N^\xi\rho^{3/2}}{N^2\eta^{3/2}}\right),
\end{split}
\end{equation}
where $\kappa^{(i)}(ab,{\bm \alpha})$ denotes the joint cumulant of the random variables $h_{ab}^i$, $h_{\alpha_1}^i, \dots, h_{\alpha_k}^i$, and  $\partial_{\bm\alpha}^{(i)}:=\partial_{\alpha_1}^{(i)}\cdots \partial_{\alpha_k}^{(i)}$, with $i=1,2$, where $\partial_{\alpha_j}^{(i)}$ denotes the directional derivative in the direction $h_{\alpha_j}^i$. Here $h_{\alpha_j}^i$ are the entries of $H_i$. Combining \eqref{eq:expa} with \eqref{eq:geqsff} we obtain exactly the expansion in \eqref{eq:expmpres} (recall that here we only present the proof in the complex case, the real case being completely analogous).

\begin{proof}[Proof of the second equality in \eqref{eq:expa}]

First of all we recall that by \eqref{mom} it follows the bound $|\kappa^{(i)}(ab,{\bm\alpha})|\lesssim N^{-(k+1)/2}$, with $i=1,2$.

We start with $k=2$. In this case we can neglect the summation when $a=b$ since it gives a contribution $N^{-3/2}$. Hence we can assume that $a\ne b$. In this case we have the bounds
\begin{equation}
\label{eq:needbk2a}
N^{-5/2}\left|\sum_{a\ne b} G_{ab}^3\right|\lesssim \frac{N^\xi\rho^{3/2}}{N^2\eta^{3/2}}, \qquad N^{-5/2}\left|\sum_{a\ne b} G_{aa}G_{bb}G_{ab}\right|\lesssim \frac{N^\xi}{N^{3/2}}+\frac{N^\xi\rho^{3/2}}{N^2\eta^{3/2}},
\end{equation}
with very high probability. The first bound in \eqref{eq:needbk2a} follows from the isotropic law in \eqref{loclaw}. The second bound in \eqref{eq:needbk2a} follows by writing $G=m+(G-m)$ and using the isotropic resummation
\begin{equation}
\sum_{ab} (G-m)_{aa}G_{ab}=\sum_a\braket{{\bm e}_a,G \bm1},
\end{equation}
with ${\bm e}_a\in\mathbf{R}^N$ the unit vector in the $a$-direction and $\bm1:=(1,\dots,1)\in\mathbf{R}^N$.

For $k=3$ whenever there are at least two off-diagonal $G$'s we get a bound $N^{-2}\eta^{-1}\rho$. The only way to get only diagonal $G$'s is that ${\bm \alpha}$ is one of $(ab,ba,ba)$, $(ba,ab,ba)$, $(ba,ba,ab)$; in this case $\kappa^{(i)}(ab,{\bm \alpha})=\kappa_4/N^2$, with $\kappa_4:=\kappa^{(i)}(ab,ba,ab,ba)$. For these terms we have (see \cite[Lemma 4.2]{2012.13218} for the analogous proof for Wigner matrices)
\begin{equation}
\label{eq:41a}
\partial_{\bm\alpha}^{(i)}G_{ba}=-2s_i^3G_{aa}^2G_{bb}^2+\mathcal{O}\left(\frac{N^\xi\rho}{N^2\eta}\right),
\end{equation}
with very high probability, where the error comes from terms with at least two off-diagonal $G$'s. Hence we finally conclude that the terms $k=3$ give a contribution:
\begin{equation}
\label{eq:42a}
-2\kappa_4\frac{3}{3!}\norm{s}_4^4\frac{1}{N^3}\sum_{ab}G_{aa}^2G_{bb}^2=\frac{\kappa_4}{N}\norm{s}_4^4 m^4+\mathcal{O}\left(\frac{N^\xi\rho^{3/2}}{N^{3/2}\eta^{1/2}}+\frac{N^\xi\rho}{N^2\eta}\right).
\end{equation}

All the terms with $k\ge 4$ can be estimated trivially using that $|G_{ab}|\lesssim 1$ with very high probability by \eqref{loclaw}.

\end{proof}

\subsection{Computation of the variance}
\label{sec:varres}

For the second moment, using \eqref{eq:singgeq}, we compute
\begin{equation}
\E_H\braket{G_1-\E_H G_1}\braket{G_2-\E_H G_2}=-\E_H\left(\frac{m_1'}{m_1}\braket{\un{s_1^1H_1 G_1}+s_2^1\un{H_2G_1}}+\frac{\kappa_4}{N}\norm{s^1}_4^4 m_1'm_1^3\right)\braket{G_2-\E_H G_2}+\mathcal{O}\left(\frac{N^\xi\Psi_2}{L^{1/2}}\right)
\end{equation}
where $s^i=(s_1^i,s_2^i)\in S^1$ and we used \eqref{eq:expmpres} to approximate $\braket{G_i-\E_H G_i}$ with $\braket{G_i-m_i}$. We made this replacement to use the equation for $G-m$ from \eqref{eq:singgeq}.

Then performing cumulant expansion we compute:
\begin{equation}
\label{eq:commpscemoma}
\begin{split}
&-\E_H\left(\frac{m_1'}{m_1}\braket{s_1^1\un{H_1 G_1}+s_2^1\un{H_2G_1}}+\frac{\kappa_4}{N}\norm{s^1}_4^4 m_1'm_1^3\right)\braket{G_2-\E_H G_2} \\
&\qquad=\frac{\braket{s^1,s^2} m_1'\E_H\braket{G_1G_2^2}}{m_1N^2}-\frac{\kappa_4}{N}\norm{s^1}_4^4 m_1'm_1^3\E_H\braket{G_2-\E_H G_2} \\
&\qquad\quad -\frac{m_1'}{m_1}\sum_{k\ge 2}\sum_{ab}\sum_{{\bm\alpha}\in\{ab,ba\}^k}\left(\frac{\kappa^{(1)}(ab,{\bm\alpha})}{k!N}s_1^1\partial_{{\bm \alpha}}^{(1)}+s_2^1\frac{\kappa^{(2)}(ab,{\bm\alpha})}{k!N}\partial_{{\bm \alpha}}^{(2)}\right)\E_H\big[(G_1)_{ba}\braket{G_2-\E_H G_2}\big].
\end{split}
\end{equation}

Using the local law~\cref{eq:23gl} we conclude that
\begin{equation}
\begin{split}
\frac{m_1'}{m_1}\braket{s^1,s^2}\frac{\braket{G_1G_2^2}}{N^2}&=\braket{s^1,s^2}\frac{m_1'm_2'}{(1-\braket{s^1,s^2}m_1m_2)^2N^2}+\mathcal{O}\left(\frac{N^\xi}{N^3\eta_1\eta_2\eta_*^2}+\frac{N^\xi}{N^4|\eta_1\eta_2|^3}\right) \\
&=-\frac{1}{N^2}\partial_{z_1}\partial_{z_2}\log(1-\braket{s^1,s^2}m_1m_2)+\mathcal{O}\left(\frac{N^\xi}{N^3\eta_1\eta_2\eta_*^2}+\frac{N^\xi}{N^4|\eta_1\eta_2|^3}\right),
\end{split}
\end{equation}
with very high probability.

We are now left with the third line of \eqref{eq:commpscemoma}. The ${\bm \alpha}$-derivative in \eqref{eq:commpscemoma} may hit either $(G_1)_{ba}$ or $\braket{G_2-\E_2 G_2}$. Define
\begin{equation}
\label{eq:almlowor}
\begin{split}
\Phi_k:&= \frac{m_1'}{m_1}\sum_{ab}\sum_{{\bm\alpha}\in\{ab,ba\}^k}\left(\frac{\kappa^{(1)}(ab,{\bm\alpha})}{k!N}s_1^1\partial_{{\bm \alpha}}^{(1)}+s_2^1\frac{\kappa^{(2)}(ab,{\bm\alpha})}{k!N}\partial_{{\bm \alpha}}^{(2)}\right)\E_H\big[(G_1)_{ba}\braket{G_2-\E_H G_2}\big]\\
&=\sum_{ab}\sum_{\bm \alpha}\frac{s_1^1\kappa^{(1)}(ab,{\bm\alpha})}{k! N}\E_H\left(\frac{m_1'}{m_1}\partial_{\bm \alpha_1}^{(1)}\frac{(G_1)_{ba}}{k_1!}\right)\left(\partial_{\bm \alpha_2}^{(1)}\frac{\braket{G_2-\E_H G_2}}{(k-k_1)!}\right)\\
&\quad+\sum_{ab}\sum_{\bm \alpha}\frac{s_2^1\kappa^{(2)}(ab,{\bm\alpha})}{k! N}\E_H\left(\frac{m_1'}{m_1}\partial_{\bm \alpha_1}^{(2)}\frac{(G_1)_{ba}}{k_1!}\right)\left(\partial_{\bm \alpha_2}^{(2)}\frac{\braket{G_2-\E_H G_2}}{(k-k_1)!}\right),
\end{split}
\end{equation}
where $k_1$ denotes the number of derivatives that hit $(G_1)_{ba}$. The summation $\sum_{\bm \alpha}$ indicates the summation over tuples ${\bm \alpha}_i^{k_i}$, with $i=1,2$ and $k_2:=k-k_1$. We now claim that
\begin{equation}
\label{eq:maineq4tha}
\Phi_k=-\bm 1(k=3)\Bigl(\kappa_4\frac{\braket{s^1\odot s^1,s^2\odot s^2}}{2N^2}(m_1^2)'(m_2^2)'+\frac{\kappa_4}{N}\norm{s^1}_4^4 m_1'm_1^3\Bigr)+\mathcal{O}\left(N^\xi\frac{\Psi_2}{L^{1/2}}\right).
\end{equation}

Similarly to the proof of~\cite[Eq.\ (113)]{2012.13218} we readily conclude that the terms
in $\Phi_k$ in~\eqref{eq:almlowor} with $k=2$, or $k_1$ odd and $k\ge 4$, or $k\ge 3$ and $k_1$ even are bounded by $N^\xi\Psi_2 L^{-1/2}$. For $k=3$ and $k_1=3$, analogously to \eqref{eq:41a}--\eqref{eq:42a} we obtain a contribution of
\begin{equation}
-\frac{\kappa_4}{N}\norm{s^1}_4^4 m_1'm_1^3+\mathcal{O}\left(\frac{N^\xi}{N|\eta_1|L^{1/2}}\right)
\end{equation}
to~\cref{eq:maineq4tha}. 

For $k=3$ and $k_1=1$ we start computing the action of the ${\bm \alpha_1}$-derivative on $(G_1)_{ba}$:
\begin{equation}
\sum_{\bm \alpha_1}\partial_{\bm \alpha_1}^{(i)}(G_1)_{ba}=-s_i^1(G_1)_{ba}^2-s_i^1(G_1)_{aa}(G_1)_{bb}=-s_i^1m_1^2(1+\delta_{ab})+\mathcal{O}\left(N^\xi\sqrt{\frac{\rho_1}{N|\eta_1|}}\right),
\end{equation}
with very high probability. Additionally, we have that (see \cite[Lemma 4.2]{2012.13218} for the analogous proof for Wigner matrices)
\begin{equation}
\partial_{ab,ba}^{(i)}\braket{G_2-\E_H G_2}=\frac{2 m_2m_2'}{N}(s_i^2)^2+\mathcal{O}\left(\frac{N^\xi\rho_2^{1/2}}{(N|\eta_2|)^{3/2}}\right),
\end{equation}
with very high probability. We thus conclude that the \((k,k_1)=(3,1)\) contribution to~\cref{eq:maineq4tha} is
\begin{equation}
\label{eq:impalmt}
-\kappa_4\frac{\braket{s^1\odot s^1,s^2\odot s^2}}{2N^2}(m_1^2)'(m_2^2)'+\mathcal{O}\left(\frac{N^\xi\Psi_2}{L^{1/2}}\right),
\end{equation}
where we used that only the terms with $\kappa_4=\kappa^{(i)}(ab,ba,ab,ba)$ contribute. This concludes the proof of \eqref{eq:wickresmain} for $p=2$.

\subsection{Asymptotic Wick Theorem}
\label{sec:secwick}

The proof of the Wick Theorem for resolvent is completely analogous to the one for Wigner matrices in \cite[Section 4]{2012.13218}. The only differences are that along the proof we have to carefully keep track of the $s_i$, as we did in Section~\ref{sec:varres}, since in the Wigner case $s^1=\dots=s^p=(1,0)$, and that we have to use the three $G$'s local law in \eqref{eq:23gl} with a weaker error term instead of the one in \cite[Eq. (45)]{2012.13218} to compute the leading order deterministic term (see \eqref{eq:2gp1}--\eqref{eq:2gp2} below).

Define
\begin{equation}
Y_{S}:=\prod_{i\in S} \braket{G_i-\E_H G_i},
\end{equation}
with $S\subset \mathbf{N}$. Similarly to Section~\ref{sec:varres} we start computing
\begin{equation}
\begin{split}
\E_H Y_{[p]}&=\sum_{i\in [2,p]}\frac{m_1'}{m_1}\frac{\braket{s^1,s^i}}{N^2}\E_H\braket{G_1G_i^2} Y_{[p]\setminus\{1,i\}}-\frac{\kappa_4}{N}\norm{s^1}_4^4 m_1'm_1^3 \E_H Y_{[2,p]} \\
&\quad-\sum_{k\ge 2}\sum_{ab}\sum_{{\bm\alpha}\in\{ab,ba\}^k}\left(\frac{\kappa^{(1)}(ab,{\bm\alpha})}{k!N}s_1^1\partial_{{\bm \alpha}}^{(1)}+s_2^1\frac{\kappa^{(2)}(ab,{\bm\alpha})}{k!N}\partial_{{\bm \alpha}}^{(2)}\right)\E_H\left[\frac{m_1'}{m_1}(G_1)_{ba}Y_{[2,p]}\right] \\
&\quad+\mathcal{O}\left(N^\xi\frac{\Psi_p}{L^{1/2}}\right).
\end{split}
\end{equation}
Then proceeding analogously to \eqref{eq:almlowor}--\eqref{eq:impalmt} (see also \cite[Eqs. (110)-(114)]{2012.13218} for the Wigner case) we conclude that
\begin{equation}
\label{eq:2gp1}
\begin{split}
\E_H Y_{[p]}&=\sum_{i\in [2,p]}\frac{m_1'}{m_1}\frac{\braket{s^1,s^i}}{N^2}\E_H\braket{G_1G_i^2} Y_{[p]\setminus\{1,i\}}\\
&\quad-\sum_{i\in [2,p]} \kappa_4 \frac{\braket{s^1\odot s^1,s^i\odot s^i}}{2N^2}(m_1^2)'(m_i^2)'\E_H Y_{[1,p]\setminus\{1,i\}}+\mathcal{O}\left(N^\xi\frac{\Psi_p}{L^{1/2}}\right).
\end{split}
\end{equation}
In order to compute the leading deterministic term of $\braket{G_1G_i^2}$ we use the local law \eqref{eq:23gl} and get
\begin{equation}
\label{eq:2gp2}
\E_H Y_{[p]}=\frac{1}{N^2}\sum_{i\in [2,p]}V_{1,i}\E_H Y_{[p]\setminus\{1,i\}}+\mathcal{O}\left(N^\xi\Psi_p\left(\frac{1}{L^{1/2}}+\frac{1}{N\eta_*^2}+\frac{1}{N^2\eta_*^4}\right)\right).
\end{equation}
Finally, proceeding iteratively we conclude \eqref{eq:wickresmain}.

\subsection{Multi resolvents local laws}
The goal of this section is to prove the local laws in \eqref{eq:23gl}. Starting from \eqref{eq:singgeq} we get
\begin{equation}
\begin{split}
(1-\braket{s^1,s^2}m_1m_2\braket{\cdot})G_1G_2&=m_1m_2+m_1\braket{G_2-m_2}-m_1\big(s_1^1\underline{H_1G_1G_2}+s_2^1\underline{H_2G_1G_2}\big) \\
&\quad+m_1\braket{s^1,s^2}\braket{G_1G_2}(G_2-m_2)+m_1\braket{G_1-m_1}G_1G_2.
\end{split}
\end{equation}
We estimate $|\braket{G_1G_2}|\lesssim N^\xi (\eta^*)^{-1}$ with very high probability, where $\eta^*:=\eta_1\vee \eta_2$, using
$|\braket{G_1G_2}| \le \braket{|G_2|}/\eta_1$ (in case $\eta^*=\eta_1$) and the rigidity of eigenvalues to estimate  $\braket{|G_2|}\le N^\xi$.
Then  by the single resolvent local law $|\braket{G_i-m_i}|\lesssim N^\xi (N\eta_i)^{-1}$ from \eqref{loclaw}
we obtain
 that
\begin{equation}
(1-\braket{s^1,s^2}m_1m_2)\braket{G_1G_2}=m_1m_2-m_1\big(s_1^1\braket{\underline{H_1G_1G_2}}+s_2^1\braket{\underline{H_2G_1G_2}}\big)+\mathcal{O}\left(\frac{N^\xi}{N|\eta_1||\eta_2|}\right),
\end{equation}
with very high probability. Finally, using that
\begin{equation}
\label{eq:2gunder}
|\braket{\underline{H_iG_1G_2}}|\lesssim \frac{N^\xi}{\sqrt{N|\eta_1\eta_2|}\eta_*}, \qquad i\in [2]
\end{equation}
with very high probability from an analogous proof to \cite[Eq. (5.8)]{2106.10200} (see also \cite[Eq. (5.10c)]{1912.04100}), and that
\begin{equation}
\label{eq:lbso}
|1-\braket{s^1,s^2}m_1m_2|\gtrsim \eta^*.
\end{equation}
we conclude the first local law in~\eqref{eq:23gl}.

For the second local law in~\eqref{eq:23gl} we start writing the equation for $G_1G_2^2$:
\begin{equation}
\begin{split}
G_1G_2^2&=m_1m_2'+m_1(G_2^2-m_2')-m_1\big(s_1^1\underline{H_1G_1G_2^2}+s_2^1\underline{H_2G_1G_2^2}\big) \\
&\quad+m_1 \braket{s^1,s^2}\big(\braket{G_1G_2}G_2^2+\braket{G_1G_2^2}G_2\big)+m_1\braket{G_1-m_1}G_1G_2^2.
\end{split}
\end{equation}
Then, using the usual single $G$ local law and the two $G$'s local law from \eqref{eq:23gl}, we conclude that
\begin{equation}
\begin{split}
(1-\braket{s^1,s^2}m_1m_2)\braket{G_1G_2^2}&=m_1m_2'+ \braket{s^1,s^2}\frac{m_1^2m_2m_2'}{1-\braket{s^1,s^2}m_1m_2} \\
&\quad -m_1\big(s_1^1\underline{H_1G_1G_2^2}+s_2^1\underline{H_2G_1G_2^2}\big) +\mathcal{O}\left(\frac{N^\xi}{N|\eta_1||\eta_2|\eta_*}\right).
\end{split}
\end{equation}
Then, using that
\begin{equation}
\label{eq:3gunder}
|\braket{\underline{H_iG_1G_2^2}}|\lesssim \frac{N^\xi}{N\sqrt{|\eta_1\eta_2|}\eta_*^2}, \qquad i\in [2],
\end{equation}
with very high probability, and \eqref{eq:lbso} we conclude \eqref{eq:23gl}. The proof of \eqref{eq:3gunder} follows analogously to the one of \eqref{eq:2gunder}.

\section{Stationary phase calculations}\label{sec:statphase}

The proof of~\eqref{EHmonSTmon} is a tedious stationary phase calculation since $v_\pm^{sr}(t)$,
the leading part of  $v_{\pm,\kappa}^{sr}(t)$ (see~\eqref{vpm kappa}),  are
given in terms of oscillatory integrals for $t\gg 1$ being the large parameter.  Unlike in the $s=r$
case, no explicit formula similar to~\eqref{v00pm} is available. The main complication is 
that $V^{sr}(x,y)$ defined in \eqref{vsrfg} has logarithmic singularities, integrated against
a fast oscillatory term from $f'g'$, so standard stationary 
phase formulas cannot directly be applied. Nevertheless, a certain number of integration by
parts can still be performed before  the derivative of the integrand stops being integrable
and the leading term can be computed.

We will first give a proof of 
\be\label{tar1}
\E_s\E_r v^{sr}_-(t) \sim \sqrt{t}
\ee
then we explain how to modify this argument to obtain
\be\label{tar2}
\E_s\E_r v^{sr}_-(t)^2 \sim t^{3/2},
\ee
in both cases with a definite large $t$  asymptotics with computable explicit constants.
The proof reveals that the corresponding results for $\E_s\E_r v^{sr}_+(t)$ and $\E_s\E_r v^{sr}_+(t)^2$ 
guarantee only an upper bound with the same behavior
\be\label{tar3}
\E_s\E_r v^{sr}_+(t) \lesssim \sqrt{t}, \qquad \E_s\E_r v^{sr}_+(t)^2 \lesssim t^{3/2}
\ee
 depending on the distribution of
$s$ on $S^1$, the matching lower bound may
not necessarily hold.  However, for our main conclusions like~\eqref{compa} only an upper bound on $S_\mathrm{res}(t)$
is important.

All these exponents are valid for the $k=2$ case, i.e.\ for $H^s=s_1H_1+s_2H_2$. For the general multivariate
model, $k\ge 3$, exactly the same proof gives the upper bounds
\be\label{tar4}
  \E_s\E_r v^{sr}_\pm(t) \lesssim \min\{ 1, t^{\frac{3-k}{2}} \}, \qquad \E_s\E_r v^{sr}_+(t)^2 \lesssim \min\{ 1, t^{\frac{5-k}{2}} \}.
\ee
The $k$-dependence of the exponent can directly be related  to the tail behavior~\eqref{rhou} and~\eqref{tail}
 below, so for simplicity we  will carry out our main analysis only for  $k=2$. 
 In fact, a  more careful analysis yields somewhat better bounds than~\eqref{tar4}, but we will not pursue this improvement here.

We introduce a new random variable
$$
   U: = \langle s, r\rangle%
$$
then clearly $|U|\le 1$ and since  $r, s\in S^k$ have a distribution with an $L^2$ density, it is easy to see
that the density $\rho^*$ of $U$ is bounded by 
\be\label{rhou}
  \rho^* (U)\lesssim  (1-U^2)^{\frac{k-3}{2}}. %
\ee
The fact that the main contribution to the lhs. of \eqref{tar4} comes from the regime $U\approx 1$ is a consequence of the singularity of the logarithm in \eqref{vsrfg} in this regime (see computations below). Indeed, $U=\cos \alpha$ where $\alpha$ is the angle between $r, s$ and near $U\approx \pm 1$
we have $1\pm U \approx \frac{1}{2}\alpha^2(1+ O(\alpha^2))$.
For example, for $k=2$  we have
\be\label{1-U}
   \Prob ( 1-U= \epsilon+\dif \epsilon) = \frac{\dif \epsilon}{\sqrt{\epsilon}} \Big(\int_{S^1}\rho^2(s)\dif s\Big)(1+O(\sqrt{\epsilon}))
\ee
\be\label{1+U}
    \Prob ( 1+U= \epsilon+\dif \epsilon) = \frac{\dif \epsilon}{\sqrt{\epsilon}} \Big(\int_{S^1}\rho(s)\rho(s+\pi)\dif s \Big)(1+O(\sqrt{\epsilon}))
\ee
in the $\epsilon\ll 1$  regime.
In particular, the bound in \eqref{rhou} is actually an asymptotics  in the most critical $U\approx 1$ regime,
while the regime $U\approx -1$ it may happen
that the density $\rho^*$ is much smaller than \eqref{rhou} predicts. For symmetric distribution, $\rho(s)=\rho(s+\pi)$,
the two asymptotics are the same.
Similar relations hold for $k\ge 3$, in which case we have
\be\label{tail}
   \Prob ( 1\pm U = \epsilon +\dif \epsilon ) %
   \lesssim \epsilon^{\frac{k-3}{2}}\dif \epsilon
\ee
with an explicit asymptotics for $U\approx 1$.

So  we will study
\be\label{RR}
  R_\pm(t) = t^2 \Re \int \dif U \rho^*(U) \iint_{-2}^2 \dif x\dif y e^{\ii t (x\pm y)} 
\Big[  \log \abs*{1-U m(x)\ov{m(y})} - \log \abs*{1-Um(x)m(y)}\Big].
\ee
Since $|m|\le 1$, as long as $|U|\le 1-\delta$ for any small fixed $\delta>0$,
the arguments of the logarithms are separated away from zero
 and they allow to perform arbitrary number of integration by parts, each gaining a factor of $1/t$.
 There is a square root singularity of $m(x)$ and $m(y)$
  at the spectral edges $2, -2$ which still allows  one to perform one integration by parts in each variable
  since $m'$  is still integrable. Therefore the contribution of the regime  $|U|\le 1-\delta$
  to~\eqref{RR} is of order $t^2(1/t)^2= O(1)$, hence negligible compared with the target~\eqref{tar1}.
In the sequel we thus focus on the important $U\approx \pm 1$ regimes, in particular
every $\int\dif U$ integral is understood to be restricted to $|U|\ge 1-\delta$.

Note that $\ov{m(y)} =-m(-y)$, so if $U$ has a symmetric distribution (for example if $s\in S^1$ has a symmetric distribution),
then by symmetry we have
$$   
    R_-(t)=-R_+(t).
$$

For definiteness, we focus on $R_-(t)$, the analysis of $R_+$ is analogous.
 From the explicit form $m(x) =\frac{1}{2}(-x+\ii \sqrt{4-x^2})$ a simple exercise shows that
 \be\label{lowerb}
   |1-Um(x)\ov{m(y)}|^2 \gtrsim (1-U)^2 + (x-y)^2, \qquad  |1-Um(x) m(y)|^2 \gtrsim (1+U)^2 + (x+y)^2.
 \ee
This shows that the critical regime is $U\approx 1$ and $x\approx y$ for the first integrand in \eqref{RR}
and $U\approx -1$, $x\approx -y$ for the second. 
Again, for definiteness, we focus on the first regime, i.e.
on the first log-integrand in \eqref{RR}  and establish the following relations for large $t$ and $k=2$:

\begin{lemma}\label{Rlemma} In the $k=2$ case we have
\be\label{R0}
t^2 \int \dif U \rho^*(U) \Re \iint_{-2}^2  e^{\ii t (x- y)} 
 \log \abs*{1-U m(x)\ov{m(y})}^2 \dif x \dif y \sim \sqrt{t}
\ee
and
\be\label{R1}
t^4  \int \dif U \rho^*(U) \Bigg[ \Re \iint_{-2}^2  e^{\ii t (x- y)} 
 \log \abs*{1-U m(x)\ov{m(y})}^2 \dif x \dif y \Bigg]^2 \sim t^{3/2},
\ee
for $t\ge 1$.  For $t\gg 1$ an analogous asymptotic statement holds
 with explicitly computable positive constants that depend on the distribution of $s$.
\end{lemma}
\begin{proof}[Proof of Lemma~\ref{Rlemma}]Introduce the variables
$$
   a:=\frac{x+y}{2}, \quad b:=\frac{x-y}{2}, \qquad \mbox{i.e.} \quad x= a+b, \quad y= a-b.
$$
Since $|x|, |y|\le 2$ we have  
\be\label{range}
|a|\le 2, \qquad |b|\le \min\{ |2-a|, |2+a|\}.
\ee
 In terms of these variables, we have
\be\label{ab}
   |1-Um(x)\ov{m(y)}|^2  = \Big( 1- U + 2U \frac{b^2}{ b^2 + d^2} \Big)^2
   +  \frac{4 U^2 b^2  d^2}{(b^2+ d^2)^2}, \qquad d:= \frac{1}{2}\big[ \sqrt{4-(a+b)^2} +\sqrt{4-(a-b)^2}\big].
\ee
Here we also used the identity
$$
   1-m(x)\ov{m(y)} = \frac{2b}{2b + m(x)-\ov{m(y)}} =\frac{2b}{b+\ii d}
$$
following from 
the equation $-m(x)^{-1}=x+ m(x)$ and  similarly  for $m(y)$.
In the regime~\eqref{range} we have 
\be\label{best}
|b|\le \frac{1}{2}(4-a^2), \qquad |b|\le \sqrt{4-a^2}.
\ee
Note that by Taylor expansion around $a$  and  concavity  of the function $x\to \sqrt{4-x^2}$  in $x\in[-2,2]$, we have
\be\label{ddef}
    0\le \sqrt{4-a^2}-d  \lesssim  \frac{b^2}{(4-a^2)^{3/2}}\le 
    \frac{|b|}{\sqrt{4-a^2}}, \quad \mbox{as well as} \quad \frac{1}{2}\sqrt{4-a^2}\le d\le \sqrt{4-a^2}.
\ee
We define  the function 
\be\label{defF}
  F = F(U, a, b): = (1-U)^2 +\frac{4U^2b^2}{4-a^2}
\ee
for $|U|\le 1$, and $a, b$ as in~\eqref{range}. We will use $F$
to approximate  
\be\label{Mdef}
  M=M(U, a,b):=|1-Um(a+b)\ov{m(a-b)}|^2 
\ee
 in the critical regime
where $|U|\ge 1-\delta$ and $|b|\le\delta$  for some small fixed $\delta>0$.
We clearly have
\be\label{appbound}
  M(U, a,b)\ge \frac{1}{4} F(U, a, b)
\ee
in the regime~\eqref{range}, where $|b|\le \sqrt{4-a^2} \le 2d$, using~\eqref{ddef}.

For the difference function
\be\label{deltadef}
  \Delta(U, a, b):=  M(U, a,b)- F(U, a, b)
\ee
an elementary calculation
from~\eqref{ab}--\eqref{ddef}  gives
\be\label{appapp}
\big| \Delta(U, a, b)\big| %
\lesssim \frac{b^2}{(4-a^2)^{3/2}} \sqrt{F}
\ee
in the regime $|U|\ge 1-\delta$ and $|b|\le\delta$.
Furthermore, similar estimates hold for the first derivative;
\be\label{1der}
  \Bigg| \frac{\dif }{\dif b}\Delta(U, a, b)\Big]\Bigg| 
  \lesssim \frac{|b|\sqrt{F}}{(4-a^2)^{3/2}} , \quad  \Bigg| \frac{\dif }{\dif a}\Delta(U, a, b)\Big]\Bigg| 
  \lesssim \frac{b^2\sqrt{F}}{(4-a^2)^{5/2}} \lesssim \frac{|b|\sqrt{F}}{(4-a^2)^{3/2}},
\ee
as well as for the second derivatives
\be\label{2der}
  \Bigg| \frac{\dif^2 }{\dif b^2}\Delta(U, a, b)\Big]\Bigg| 
  \lesssim  \frac{\sqrt{F}}{(4-a^2)^{3/2}} , \quad 
  \Bigg|  \frac{\dif }{\dif a}\frac{\dif }{\dif b}\Delta(U, a, b)\Big]\Bigg| 
  \lesssim  \frac{|b|\sqrt{F}}{(4-a^2)^{5/2}} \lesssim \frac{\sqrt{F}}{(4-a^2)^{3/2}} .
\ee

The proof of Lemma~\ref{Rlemma} consists of two parts. First we compute
the  integral with $\log F$, i.e.\ we show that
\be\label{intapp2}
t^2  \int \dif U \rho^*(U) \Re \iint_{-2}^2  e^{\ii t (x- y)} 
\log F\big(U, \frac{x+y}{2},  \frac{x-y}{2}\big)\dif x \dif y \sim  \sqrt{t}
\ee
with an explicit positive constant factor in the asymptotic regime $t\gg 1$.
Second, we show that the integrand
in~\eqref{R0} can indeed be replaced with $F$ up to a negligible error,
\be\label{intapp}
\Bigg| t^2 \int \dif U \rho^*(U) \iint_{-2}^2  e^{\ii t (x- y)} \Big[ 
 \log \abs*{1-U m(x)\ov{m(y})}^2 - \log F\big(U, \frac{x+y}{2},  \frac{x-y}{2}\big)\Big] \dif x \dif y\Bigg| \lesssim 1.
\ee

{\bf Part I.} 
To prove~\eqref{intapp2}, we use the $a, b$ variables and  the symmetry of $F$ in $a$ to restrict
the $a$ integration to $0\le a\le 2$:
\be\label{rest}
\eqref{intapp2} =4t^2 \Re \int \dif U \rho^*(U) \int_0^2 \dif a  \int_{-(2-a)}^{2-a}  \dif b\; e^{2\ii t b} 
\log F\big(U,a,b) .
\ee
Using integration by parts, we have
\be\label{intpp}
\begin{split}
\int_{-(2-a)}^{2-a}  \dif b\; e^{2\ii t b} \log \big[ (1-U)^2 +\frac{4U^2b^2}{4-a^2}\big] 
 =  & \frac{1}{2\ii t} \Big[ e^{2\ii t (2-a)} - e^{-2\ii t (2-a)}\Big] \log \big[ (1-U)^2 +\frac{4U^2(2-a)}{2+a}\big] \\
 &- \frac{1}{2\ii t} \frac{4U^2}{4-a^2} \int_{-(2-a)}^{2-a}  \dif b\;  e^{2\ii t b} \frac{ 2b} { (1-U)^2 +\frac{4U^2b^2}{4-a^2} }.
 \end{split}
 \ee
In the boundary terms we can perform one more integration by parts in the $a$ variable
when plugged into~\eqref{rest}. Just focusing on the
first boundary term in~\eqref{intpp}, using $ |U|\le 1$ we have
$$
 \Bigg| \frac{1}{2\ii t}  e^{4\ii t}  \int_0^2 \dif a \, e^{-2\ii t a}\log \big[ (1-U)^2 +\frac{4U^2(2-a)}{2+a}\big] \Bigg|
 \lesssim \frac{1}{t^2} \int_0^2  \frac{\dif a}{ (1-U)^2 + U^2(2-a)} \lesssim \frac{|\log (1-U)|}{t^2}.
$$
Since $\rho^*(U)$ is a density bounded  by $(1-U^2)^{-1/2}$ in the $U\approx 1$ regime from~\eqref{rhou},  the
logarithmic singularity is integrable showing that the two boundary terms in \eqref{intpp}, when  plugged into~\eqref{rest},  give at most
an $O(1)$ contribution, negligible compared with the target behavior of order $\sqrt{t}$ in~\eqref{tar1}.

To compute the main (second) term in the rhs. of~\eqref{intpp}, we first extend the integration limits to infinity and claim that
\be\label{onemore}
\begin{split}
 t^2  \int \dif U \rho^*(U) \Big|  \frac{1}{2\ii t} &\int_0^2\dif a \frac{4U^2}{4-a^2} \int_{2-a}^{\infty}  \dif b\;  e^{2\ii t b} \frac{ 2b} { (1-U)^2 +\frac{4U^2b^2}{4-a^2} }\Big|\\
& \lesssim   t \int \dif U \rho^*(U)\int_0^2\frac{\dif a}{2-a}\Bigg|
 \int_{2-a}^{\infty}  \dif b\;  e^{2\ii t b} \frac{ 2b} { (1-U)^2 +\frac{4U^2b^2}{4-a^2} }\Bigg|
 \end{split}
\ee
gives a negligible contribution to~\eqref{rest} (the lower limit is removed similarly).
Indeed, 
 we apply one more integration by parts inside the absolute value in~\eqref{onemore}:
 $$
\Bigg|
 \int_{2-a}^{\infty}  \dif b\;  e^{2\ii t b} \frac{ 2b} { (1-U)^2 +\frac{4U^2b^2}{4-a^2} }\Bigg|
 \lesssim t^{-1} \int_{2-a}^{\infty} \frac{\dif b}{  (1-U)^2 +\frac{U^2b^2}{4-a^2} } + t^{-1} \frac{ 2-a } { (1-U)^2 +(2-a) }.
$$
Its contribution to  the rhs of~\eqref{onemore} is thus bounded by
$$
  \int \dif U \rho^*(U) \int_0^2\frac{\dif a}{2-a} 
  \Big[ \int_{2-a}^{\infty} \frac{\dif b}{  (1-U)^2 +\frac{U^2b^2}{4-a^2} } + \frac{ 2-a } { (1-U)^2 +(2-a) }\Big]
$$
$$
 \lesssim  \int \frac{\dif U}{\sqrt{1-U^2}}\Bigg[ \int_0^2
 \frac{\dif a}{\sqrt{2-a} }\frac{1}{1-U + \sqrt{2-a}}  + |\log(1-U)|\Bigg] \lesssim1.
 $$

Summarizing, we just proved that
\be\label{summ}
\begin{split}
\eqref{intapp2} & = -2t \Im \int \dif U \rho^*(U) \int_0^2 \dif a
\frac{4U^2}{4-a^2} \int_{-\infty}^\infty \dif b\;  e^{2\ii t b} \frac{ 2b} { (1-U)^2 +\frac{4U^2b^2}{4-a^2} }
+O(1)\\
 &= \frac{t}{\pi} \int \dif U \rho^*(U) \int_0^2 \dif a \; e^{-t\sqrt{4-a^2}(1-U)/U}
+O(1)\\
& = \frac{c_0t}{\pi}    \int \dif U \frac{1}{\sqrt{1-U}} \int_0^2 \dif a \; e^{-t\sqrt{4-a^2}(1-U)/U}  +O(1)\\
&  =  \frac{c_0\sqrt{t}}{\pi}  \int_0^\infty \frac{e^{-v}}{\sqrt{v}} \dif v \int_0^2 \frac{\dif a}{(4-a^2)^{1/4}}
    +O(1) \\
 &=  \frac{\Gamma(3/4)}{\sqrt{2}\Gamma(5/4)}c_0 \sqrt{t} + O(1),
 \end{split}
 \ee
where in the second line we used residue calculation, in the third line we used that
$$
   \rho^*(U) = \frac{c_0}{\sqrt{1-U}} +O(1)
$$
in the regime $U\approx 1$ with some positive constant $c_0>0$ depending on the distribution of $s$
(see~\eqref{1-U}),
and finally in the fourth line we used that  for large $t$ 
the main contribution to the integral comes from   $U\approx 1$
in order to simplify the integrand.  This completes the proof of~\eqref{intapp2}.

\bigskip

{\bf Part II.} We now prove~\eqref{intapp}.  After changing to the $a, b$ variables and considering only the
$0\le a\le 2$ regime for definiteness, we perform 
an integration by parts in $b$ that gives 
\be\label{intapp4}
\begin{split}
  \eqref{intapp} \lesssim & \;  t \int \dif U \rho^*(U) \Bigg| \int_{0}^2 \dif a \, e^{2\ii ta} \Big[ 
 \log M(U, a,b) -\log F\big(U, a, b\big)
 \Big]  \dif b\Bigg|  \\
&+   t \int \dif U \rho^*(U) \int_{0}^2 \dif a \Bigg| \int_{-(2-a)}^{2-a} e^{2\ii tb} \partial_b\Big[ 
 \log M(U, a,b) -\log F\big(U, a, b\big)
 \Big]  \dif b\Bigg| 
\end{split}
\ee
recalling the definition of $M$ from~\eqref{Mdef}.
 The first term  in~\eqref{intapp4}
is the boundary term, which is negligible  after one more integration by parts using the $\partial_a$ 
derivative estimate from~\eqref{1der}.
 
In the second term we perform one more integration by parts to obtain
 \be\label{intapp5}
 \begin{split}
  \eqref{intapp} \lesssim   & \; t \int \dif U \rho^*(U) \Bigg| \int_{0}^2 \dif a \, e^{2\ii ta} \partial_b\Big[ 
 \log M(U, a,b) -\log F\big(U, a, b\big)
 \Big]  \dif b\Bigg|  \\
 &+\int \dif U \rho^*(U) \int_{0}^2 \dif a \int_{-(2-a)}^{2-a}\Bigg| \partial_b^2\Big[ 
 \log M \big(U, a, b\big)-\log F\big(U, a, b\big)
 \Big] \Bigg| \dif b ,
\end{split}
\ee
where the first term comes from the boundary.
In this  term we can perform one more integration by parts in $a$. The corresponding boundary
terms are easily seen to be order one and the main term is analogous to the first term in the rhs of~\eqref{intapp5}
just we have the mixed $\partial_a\partial_b$ derivative.
 Recalling $\Delta= M-F$  from~\eqref{deltadef}, we use the estimate
 $$
  \big| \partial_b^2 [\log M- \log F]\big| \lesssim \frac{|\partial_b^2 \Delta| }{F} +  \frac{|\partial_b^2 F|}{F^2} |\Delta|+
    \frac{|\partial_b \Delta||\partial_b M +\partial_b F| }{F^2} + (\partial_bF)^2\frac{|\Delta|}{F^3}
$$
in the situation where 
  $M\gtrsim F>0$ are positive functions  (see~\eqref{appbound}). Similar bound holds for the mixed derivative.
  
  Therefore,    we can estimate 
  both  integrals in~\eqref{intapp5} as follows:
  \be\label{intapp6}
  \begin{split}
  \eqref{intapp} \lesssim & \int \dif U \rho^*(U) \int_{0}^2 \dif a \int_{-(2-a)}^{2-a}
   \frac{1}{(4-a^2)^{3/2}}  \frac{1}{\big[ (1-U)^2 + \frac{b^2}{4-a^2}\big]^{1/2}}  
  \dif b 
 \\
 \lesssim  & \int\frac{\dif U}{\sqrt{1-U^2}} \int_{0}^2 \frac{\dif a}{4-a^2} \int_0^{\sqrt{2-a}} 
    \frac{\dif u}{ \big[  (1-U)^2 + u^2\big]^{1/2}}\\
  \lesssim & \int\frac{\dif U}{\sqrt{1-U^2}}  \int_0^{\sqrt{2}} 
    \frac{ |\log u|+1 }{ \big[  (1-U)^2 + u^2\big]^{1/2}} \dif u\\
 \lesssim &\int\frac{\dif U |\log(1-U)|^2}{(1-U)^{1/2}} \lesssim 1.
 \end{split}
 \ee
 Here we used the bounds~\eqref{appapp}, \eqref{1der} and \eqref{2der} and that $|b|\le 2-a\lesssim 4-a^2$
to simplify some estimates. For computing the derivatives of $F$ we used its explicit form~\eqref{defF}.
This completes the proof of~\eqref{intapp} and thus also the proof of~\eqref{R0} in Lemma~\ref{Rlemma}.

The proof of~\eqref{R1} is very similar.  We again approximate $M=|1-Um(x)\ov{m(y)}|^2$ by
$F$ at the expense of negligible errors. We omit these calculations as they are very similar
those for~\eqref{R0} and focus only on the main term 
which is (see the analogous~\eqref{rest})
\be\label{sq}
16t^4  \int \dif U \rho^*(U) \Big[ \Re \int_0^2 \dif a  \int_{-(2-a)}^{2-a}   \dif b \; e^{2\ii tb}  \log F(U, a, b)\Big]^2.
 \ee
After one  integration by parts and neglecting the lower order boundary terms, we have the following analogue of~\eqref{summ}:
\be\label{sq2}
\begin{split}
4t^2  \int \dif U \rho^*(U) &\Big[ \Re \int_0^2 \dif a
\frac{U^2}{4-a^2} \int_{-\infty}^\infty \dif b\;  e^{2\ii t b} \frac{ 2b} { (1-U)^2 +\frac{4U^2b^2}{4-a^2} }
\Big]^2\\
 &= \frac{t^2}{\pi^2} \int \dif U \rho^*(U) \Big[ \int_0^2 \dif a \; e^{-t\sqrt{4-a^2}(1-U)/U}\Big]^2 \\
 & \approx\frac{c_0 t^{3/2}}{\pi^2}  \int_0^\infty \frac{\dif v}{\sqrt{v}} 
 \Big( \int_0^2 \dif a \; e^{-\sqrt{4-a^2}v } \Big)^2 \\
 &=  \frac{c_0 t^{3/2}}{\pi^2} \iint_0^2 \frac{\dif a_1\dif a_2}{(\sqrt{4-a_1^2}
 +\sqrt{4-a_2^2})^{1/2}}
    \int_0^\infty \frac{e^{-v}}{\sqrt{v}} \dif v \sim t^{3/2}
 \end{split}
 \ee
as the leading term. This proves~\eqref{R1} and completes the proof of Lemma~\ref{Rlemma}.
\end{proof}

 We close this section by commenting on the proof of the upper bound in~\eqref{tildeS}. Recall from~\eqref{ES} that  
 the essential part of $\wt S_\mathrm{res}(t)$  in the slope regime is given by $\E_s \E_r \wt v^{sr}(t)$
 expressed by the oscillatory integrals
 \be\label{osc}
 R_\pm(t):=t^2  \iint_{\R^2} \rho(s)\rho(r)\dif s \dif r \iint_{-2}^2 \dif x\dif y e^{\ii t (\|s\| x\pm \| r\|y)} 
A(U, x, y)
 \ee
 with 
 $$  A(U, x, y): =  \log \abs*{1-U m(x)\ov{m(y})} - \log \abs*{1-Um(x)m(y)},
 $$
 where $U= \frac{\langle s, r\rangle}{\| s\|\|r \|}$ is the cosine of the angle between the vectors $s,r\in \R^2$.
 Assuming for the moment that $\rho$, the density of $s$, is rotationally symmetric, 
 $\rho(s)= \rho(\| s\|)$ with a slight abuse of notations, we have
\be\label{EEv}
\begin{split}
 R_\pm(t) \sim & \; t^2 \int_{-1}^1 \frac{\dif U}{\sqrt{1-U^2}} \iint_{-2}^2 \dif x\dif y A(U, x, y)
 \int_0^\infty  e^{\ii t x \sigma}\rho(\sigma)\sigma \dif\sigma
 \int_0^\infty e^{\pm\ii t y \sigma'}\rho(\sigma')\sigma' \dif\sigma' \\
 \sim & \; t \int_{-1}^1 \frac{\dif U}{\sqrt{1-U^2}} \iint_{-2}^2 \dif x\dif y \wh\rho(tx) \wh{\rho(\sigma)\sigma}(\pm ty)
    \frac{\dif }{\dif x}    A(U, x, y)\\
\end{split}
\ee
performing an integration by parts in $x$ and ignoring lower order boundary term.
 In the last step
we also computed the Fourier transform (we used that  $\rho(0)=0$ to extend $\rho$ to $\R$). The main contribution comes 
from the regime where $A$ is nearly singular, and considering~\eqref{lowerb}, we just focus on the 
regime $U\sim1$  and $x\sim y$, the singularity from the other logarithmic term is treated analogously.
Similarly to the proof of~\eqref{intapp} we may ignore the edge regime, and effectively we have 
\be\label{dA}
 \Big| \frac{\dif }{\dif x}   \log \abs*{1-U m(x)\ov{m(y})} \Big| \lesssim \frac{1}{(1-U)+ |x-y|}.
\ee
Thus we can continue estimating the last line of~\eqref{EEv}
$$
  |\eqref{EEv}|\lesssim
  t \int_{-1}^1 \frac{\dif U}{\sqrt{1-U^2}} \iint_{-2}^2 \dif x\dif y \frac{\big| \wh\rho(tx) \wh{\rho(\sigma)\sigma}(ty)\big| }{ (1-U)+ |x-y|} \lesssim t^{-1/2}.
   $$
  Here we used  the regularity of  $\rho$, so that the last two factors essentially restrict the integration to the regime $|x|, |y|\lesssim 1/t$.
The final inequality is obtained just by scaling.    

To understand $\wt S_\mathrm{res}(t)$ in the ramp regime, we need to  compute $\E_s \E_r \wt v^{sr}_\pm(t)^2$, i.e.\ integrals
 of the following type:
\be\label{fin}
\begin{split}
t^4  \iint_{\R^2} & \rho(s)\rho(r)\dif s \dif r \Bigg| \iint_{-2}^2 \dif x\dif y e^{\ii t (\|s\| x\pm \| r\|y)} A(U, x, y)
\Bigg|^2\\
  = &\; t^4 \int_{-1}^1 \frac{\dif U}{\sqrt{1-U^2}}  \iint_{-2}^2 \dif x\dif y \iint_{-2}^2  \dif x'\dif y'  
   A(U, x, y)\ov{A(U, x', y')} \\
   &\times \int_0^\infty  e^{\ii t (x-x') \sigma}\rho(\sigma)\sigma \dif\sigma
 \int_0^\infty e^{\pm\ii t (y-y') \sigma'}\rho(\sigma')\sigma' \dif\sigma'\\
 \sim &\; t^2 \int_{-1}^1 \frac{\dif U}{\sqrt{1-U^2}}  \iint_{-2}^2 \dif x\dif y \iint_{-2}^2  \dif x'\dif y'  
   \frac{\dif }{\dif x}  A(U, x, y) \frac{\dif }{\dif y'} \ov{A(U, x', y')} \\
   &\times \int_0^\infty  e^{\ii t (x-x') \sigma}\rho(\sigma)\dif\sigma
 \int_0^\infty e^{\pm\ii t (y-y') \sigma'}\rho(\sigma') \dif\sigma'  \\ %
 \sim &\;  t^2 \int_{-1}^1 \frac{\dif U}{\sqrt{1-U^2}} \iint\!\!\!\iint_{-2}^2 \dif x\dif y    \dif x'\dif y'  \wh \rho(t(x-x'))\wh \rho(\pm t(y-y'))
   \frac{\dif }{\dif x}  A(U, x, y) \frac{\dif }{\dif y'} \ov{A(U, x', y')}.
 \end{split}
\ee
Here we performed two integrations by parts in $x$ and $y'$ and ignored the boundary terms. 
Estimating the derivative of $A$ as in~\eqref{dA}, we can continue 
$$
  |\eqref{fin}|\lesssim
  t^2 \int_{-1}^1 \frac{\dif U}{\sqrt{1-U^2}}  \iint_{-2}^2 \frac{\dif x\dif y }{(1-U)+ |x-y|}   \iint_{-2}^2 \frac{\dif x'\dif y' }{(1-U)+ |x'-y'|}
   \big|\wh \rho(t(x-x'))\wh \rho(t(y-y'))\big|.
   $$
The last two factors essentially restrict the integration to the regime $|x-x'|\lesssim 1/t$,
$|y-y'|\lesssim 1/t$ and by scaling we obtain a bound of order $t^{1/2}$ for $ |\eqref{fin}|$.
This completes the sketch of the proof of~\eqref{tildeS}
in the radially symmetric case, the general case is analogous but technically more cumbersome and we omit the 
details.

\emph{Data availability statement.} All data generated or analysed are included in this published article.

\end{document}